\documentclass[reqno]{amsart}
\usepackage{amssymb}
\usepackage[dvips]{epsfig}
\usepackage{graphicx}
\usepackage{color}
\usepackage{amsmath}
\usepackage{amssymb}
\usepackage{amsfonts}
\usepackage{amsthm}
\usepackage{mathrsfs}
\allowdisplaybreaks
\newcommand{\R}{\mathbb R}

\newcommand{\Z}{\mathbb Z}

\newcommand{\C}{\mathbb C}

\newcommand{\N}{\mathbb{N}}

\newtheorem{thm}{Theorem}[section]
\newtheorem{lem}[thm]{Lemma}

\newtheorem{cor}[thm]{Corollary}
\theoremstyle{remark}
\newtheorem{rem}{\bf Remark}[section]
\theoremstyle{definition}
\newtheorem{defn}[thm]{Definition}

\numberwithin{equation}{section}
\usepackage[colorlinks=true,pdfstartview=FitV,linkcolor=magenta,citecolor=cyan]{hyperref}

\usepackage{bm}
\begin{document}

	\title[]{Localization of interacting random particles with power-law long-range hopping}
	\author[W. Jian]{Wenwen Jian}
	\address[WJ] {School of Mathematics, Physics and Statistics,
		Shanghai Polytechnic University,
		Shanghai 201209,
		China}
	\email{wwjian@sspu.edu.cn}
	
	\author[Y.Sun]{Yingte Sun*}
	\thanks{*Corrseponding author}
	\address[YS]{School of Mathematical Sciences,
		Yangzhou University,
		Yangzhou 225009,
		China}
	\email{sunyt@yzu.edu.cn}

	\date{\today}

	\begin{abstract}
		In this paper, we study the interacting random particles with power-law long-rang hopping. Via the multi-scale analysis arguments for the Green's function, we establish the power-law localization for all energy with strong disorder.
	\end{abstract}
	\maketitle

	\maketitle
	\section{Introduction}
	The  configuration space of the $N$-particle lattice system is the Cartesian product $\mathbb{Z}^d\times \cdots \times \mathbb{Z}^d$, which we denote for brevity by $\mathbb{Z}^{Nd}$. The $N$-particle Hamiltonian $\mathbf{H}^{(N)}_\omega$ is a lattice operator acting on $ \bm{\phi} \in  \ell^2(\mathbb{Z}^{Nd})$: 
	\begin{equation}\label{n}
		\begin{split}
			\mathbf{H}^{(N)}_{\omega}\bm{\phi}(\textbf{x})&=\frac{1}{g}\Big(\mathbf{T}+\mathbf{U(\textbf{x})}\Big)\bm{\phi}(\textbf{x}) +\mathbf{V}(\textbf{x},\omega)\bm{\phi}(\textbf{x})\\
			& =\frac{1}{g}\Big(\sum_{\mathbf{y} \in \mathbb{Z}^{Nd}}\mathbf{T}(\textbf{x},\textbf{y})\bm{\phi}(\mathbf{y})+\mathbf{U}(\mathbf{x})\phi(\textbf{x})\Big)+\mathbf{V}(\textbf{x},\omega)\bm{\phi}(\mathbf{x}),\\
		\end{split}
	\end{equation}
	where 
	\begin{equation}
		\begin{split} &\textbf{x}=(x_1,\cdots,x_N), \quad  \textbf{y}=(y_1,\cdots,y_N) \in \mathbb{Z}^{Nd},\\
			&\mathbf{V}(\textbf{x},\omega)=\sum^N_{j=1}V(x_j,\omega).
		\end{split}
	\end{equation}

	$\bullet$ Here $x_j=\left(x^{(1)}_j,\cdots,x^{(d)}_j\right)$, $y_j=\left(y^{(1)}_j,\cdots,y^{(d)}_j\right)$ stand for coordinate vectors of the $j$th particle in  $\mathbb{Z}^d$, $j=1,\cdots,N$.
	
	Fix $1\leq n \leq N$, $\|\cdot\|$ denotes the max-norm: for $\textbf{u}=(u_1,\cdots,u_n) \in \mathbb{Z}^{nd}$,
	\begin{equation*}
		\|\textbf{u}\|=\max_{1\leq j\leq n} \|u_j\|,
	\end{equation*}
	where, for $u_j=(u^{(1)}_j,\cdots,u^{(d)}_j ) \in \mathbb{Z}^d$,
	\begin{equation*}
		\|u_j\|=\max_{i=1,\cdots,d }|u^{(i)}_{j}|. 
	\end{equation*}

	The present paper aims to establish the \textbf{power-law localization} of the random $N$-particle Hamiltonian operator $\mathbf{H}^{(N)}_{\omega}$ \eqref{n} , in which the operator $\mathbf{T}$ is  the power-law long range hopping  (see the \textbf{Assumption A} for more details).

	For the single quantum particle, Anderson localization is an well-known phenomenon  in condensed matter physics, in which particles living on a lattice become exponentially localized under random fields. The rigorous mathematical interpretation of Anderson localization can be traced back to the pioneering works \cite{FS1983,FMSS1985} of Frohlich-Spencer-etc on the multi-scale analysis (MSA) for Green's functions and the highly innovative works  \cite{AM1993,A1994} of Aizenman-Molchanov on the estimation of Green functions using fractional moment methods (FMM). We also recommend the recent books \cite{AM2009,K2008} of Aizenman-Molchanov and  Kirsch for an overview of the state-of-the-art.
	
	However, it is almost impossible to make a closed single particle in the real physical world. For a system of quantum particles with positive density\footnote{The number of particles is proportional to the system’s size.}, that we call it \textbf{many-body system}, interactions between particles are unavoidable. The central question of  ``Many-body localization" (MBL) is what happens to the localization properties of such quantum systems in the presence of interactions. This problem is envisioned in Anderson's pioneer paper \cite{A1958}.  In recent years, the research on MBL becomes a hot topic in condensed matter physics in order to create some quantum states to violate the ``Eigenstate-Thermalization Hypothesis".  Establishing a mathematical basis for MBL is ambitious and outside the scope of this paper. We shift our attention to the relatively simple case where the total number of particles is fixed. The localization of such $N$-particle system in a random environment can be obtained by strengthening the MSA method \cite{CS2009,CS20091} and FMM \cite{AM2009}, as long as  $N$ is fixed. To better distinguish from many-body system, we will refer to such systems with  finite number of quantum particles as \textbf{multi-body system}.
	
	In the aforementioned results for multi-body systems, only the ideal case is considered, where the hopping  is present only at the nearest sites. The case of long-range hopping, in particular power-law long-range, is not taken into account.  The localization properties are well established for a single particle in the presence of a random field with a power-law long-range hopping. In the influential  work  \cite{AM1993}, Aizenman-Molchanov firstly proved the power-law localization of single particles via the fractional moment method, which require the random filed to be absolute continuous. Shi \cite{Shi2021}  gave a multi-scale coupling lemma by introducing some estimations in \cite{BB13}, and proved the power-law localization under H\"older continuous random fields. A fascinating question is what happens to the localization properties of multi-body systems with power-law long-range hopping, which is the main concern of this paper.

	\subsection{Main Results}\
	
	Before stating our main localization results, we need some reasonable assumptions on the power-law long-range hopping $\mathbf{T}$, the interacting term $\mathbf{U}$ and the random potential $\mathbf{V}$.
	
	$\textbf{Assumption A:}$  The long-range hopping $\mathbf{T}$ is defined as
	\begin{equation}\label{T}
		\mathbf{T}(\textbf{x},\textbf{y})=\left\{	\begin{array}{cc}
			\frac{1}{\langle y_j-x_j \rangle^r},\ \ \ \quad   \text{if} \ x_i=y_i,\ \forall i \neq j, \\
			\\	0,\ \ \ \ \  \quad \quad  \  \text{otherwise}.
		\end{array}\right.
	\end{equation}
	where $\langle x \rangle =\max\{1,\|x\|\}$.
	\\
	
	$\textbf{Assumption B:}$  The interaction potential $\mathbf{U}$ is assumed to  be of the form
	\begin{equation*}
		\mathbf{U}(\textbf{x})=\sum_{1\leq j_1 < j_2 \leq N}U(x_{j_1},x_{j_2}). 
	\end{equation*}
	where the function $U:\ \Z^d\times\Z^d\rightarrow\R$ satisfies the following properties 
	\begin{itemize} 
		\item $U$ is a bounded   symmetric function: $\forall x,x' \in \mathbb{Z}^d, $
		\begin{align}\label{u1}
			\sup[|U(x,x')|]\leq M_1<\infty, \quad  U(x,x')=U(x',x), 
		\end{align}	
		
		\item $U$ obeys
		\begin{align}\label{u2}
			U(x,x')=0,\quad \text{if} \ \|x-x'\|\geq \mathrm{r}_0.
		\end{align}
		Here $\mathrm{r}_0\in[1,+\infty)$ is a given value (the interaction range).
	\end{itemize}

	$\textbf{Assumption C:}$  The random external potential $V(x, \omega),\ x\in\Z^d$, is assumed to be real independent identically distributed (IID) random variables, with the common probability distribution $\mu$ on some probability space $(\Omega,\mathcal{F},\mathbb{P})$. 
	
	\begin{itemize}
		
		\item The common distribution $\mu$ has compact support: For $\text{supp}(\mu)=\{x:\mu(x-\epsilon,x+\epsilon)>0 \  \text{for any } \epsilon>0\}$,
		\begin{align*}
			\text{supp}(\mu)\subset[-M,M],\ \ \ 0<M<\infty.
		\end{align*}	
		
		\item The distribution probability measure $\mu$ is  H\"older continuous of order $\rho>0$:
		\begin{align*}
			\frac{1}{\mathcal{K}_{\rho}(\mu)}=\inf_{\kappa>0}\sup_{0<|a-b|\leq\kappa}|a-b|^{-\rho}\mu([a,b])<\infty.
		\end{align*}
	\end{itemize}
	\begin{rem}
		Let $\mu$ be H\"older continuous of order $\rho>0$. Then  for any $0<\kappa<\mathcal{K}_{\rho}(\mu)$, there is some $\kappa_0=\kappa_0(\kappa,\mu)>0$ so that 
		\begin{equation*}
			\mu([a,b]) \leq \kappa^{-1}|a-b|^{\rho} \ \ \text{for}\ \ 0\leq b-a \leq \kappa_0.
		\end{equation*}
	\end{rem}
	
	The main assertion of this paper is 
	\begin{thm}\label{mainthm} Fix $N\geq 2$, and 
		consider the $N$-particle random Hamiltonian $\textbf H^{(N)}_{\omega}$ given by (\ref{n}). Suppose that  \textbf{Assumptions} $\textbf{A}$, $\textbf{B}$ and $\textbf{C}$ hold true. Let 
		\begin{equation}
			r>  \frac{40 \cdot 18^N \cdot 20Nd }{9 \rho}+\frac{75}{9}Nd.
		\end{equation}
		Then there exists $g^*=g^*(N,d,\kappa,\rho,r,\mathrm{r}_0,M,M_1)\in(0,+\infty)$ such that for any $g$ with $|g|\geq g^*$, with $\mathbb{P}$-probability one, the operator $\textbf H^{(N)}_{\omega}$ is pure point spectrum. Furthermore, there exists a complete system of eigenfunctions $\bm{\phi}_{j}(\textbf{x},\omega)$ of $\textbf H^{(N)}_{\omega}$ satisfying
		\begin{align}
			|\bm{\phi}_{j}(\textbf{x},\omega)|\leq C_j(\omega)\langle\textbf{x}\rangle^{-\frac{r}{300}}.
		\end{align}
	\end{thm}	
	
	\subsection{Some Remarks}\
	
	$\bullet$ \textbf{Main ingredients of our proof.}
	
	As mentioned above, this paper is concerned with whether the localization properties of multi-particle systems can be maintained in the presence of interaction. Absence of interaction, the multi-particle system can be decompose into independent single-particle systems, which is no different from Anderson localization.
	
	Our proof is based on the MSA for the Green's function with power-law decay. Such a proof is perturbative, and we treat the hopping and interaction  as perturbations of the random potential. The application of MSA method to $N$-particle system \eqref{n} will face two mathematical difficulties. One is that in the multi-particle phase space, the random variables on the disjoint cubes are no longer completely independent of each other, but are correlated in some sense. This makes the consideration of resonances more complicated when considering the finite volume version of Hamiltonian operator \eqref{n} on two disjoint cubes. To solve this problem, we use some key observations in \cite{CS2008,CS2009} to divide the $n$-particle cubes into partially interactive and fully interactive cases (see definition \eqref{pifi}), and only consider the pairs  of the separable cubes (see definition \eqref{spear}). By this careful partitioning, we find that in most cases the random potentials on a pair of separable cubes are independent of each other. So we can get over this problem.
	Second, since we are considering the Green's function with power-law decay, the number of iterations can only be finite if we estimate the Green's function on a larger cube by using the geometric resolvent equation in \cite{K2008}. 
	To this end, we introduce the key coupling lemma formulated by Shi \cite{Shi2021} in considering single-particle systems, that is, if the resonance is sparse, the power-law decay rate of Green's function on a larger cubes does not slow down compared to that on smaller cubes.
	
	By combining the above two methods, which of course require very precise estimates, we are able to complete the main proof in this paper.
	
	$\bullet$ \textbf{Relevant results of  quasi-periodic operator.}
	
	Moving beyond the  Anderson model, localization of quasi-periodic operators has also been well developed in the last decades. After the pioneering work of Fr\"ohlich-Spencer-Wittwer \cite{FSW1990}, the MSA method has been developed  in the study of quasi-periodic operators.   Bourgain-etc  \cite{BGS2002,B2005,B2007} greatly developed MSA method to the high-dimensional quasi-periodic Schr\"odinger operators by combing some analytical method. Recently,  Cao-Shi-Zhang \cite{CSZ2023,CSZ2004,CSZ20041} have developed the Wegner type estimate for some high-dimensional quasi-periodic Schr\"odinger operators, in which can be seen a substantial extension of MSA method in the Anderson model. In addition, there are some remarkable works related to the conclusions of this paper. Bourgin-Kachkovskiy \cite{BK2019} established localization for two interacting quasi-periodic particles by assuming symmetry and asymmetry of the potentials. Shi has carried out some in-depth studies of quasi-periodic operators with power-law long-range hopping. By introducing the Nash-Moser iterative diagonalization method,  Shi \cite{Shi2021,Shi2023} explored the optimal case of power-law localization for a class of monotone quasi-periodic operators. More recently, Shi-Wen \cite{SW2024} further developed Green's function estimates  to explore the power-law localization properties for a broader class of quasi-periodic operators.
	
	$\bullet$ \textbf{Further progress of many-body system.}
	
	We emphasize that MBL  does not constrain the number of particles. Although the interacting terms (some times including hopping terms) are also treated as perturbations, the magnitude of the perturbation does not depend on the number of particles.  This makes it materially different from the conclusions obtained in this paper. For our proof,  the magnitude  of the perturbation approaches zero rapidly as the number of particles increases.
	
	Until now, the mathematical study of MBL has been in its infancy, and the general approach to building MBL is still shrouded in fog.  The few mathematical results are restricted to one-dimensional quantum spin models: the XY chain, the quantum Ising chain, and the XXZ chain.  Hamza-Sims-Stolz \cite{HSS2012} transformed the disordered XY spin chains into  non-interacting fermions via the Jordan transformation and  the  exponential dynamical localization, uniformly for  time, is obtained  . Imbrie \cite{I2016} obtained the exponential decay of the two-point function  for the all energy of the quantum Ising chains by using KAM-type iterations. Elgart-Klein-Stolz \cite{EKS2018,EKS2019} studied a family of ferromagnetic XXZ chains, and obtained the localization properties in the low-energy region. More recently, without reduction the XXZ chains to  Schr\"odinger operators, the localization of any fixed energy interval from the bottom  has been established  by using the FMM in \cite{EK2022}.

	\subsection{Organization of the Paper}\
	
	The rest of the paper is organized as follows. In Section 2, some important facts of $n$-particle cubes and main part of  MSA argument for the Green's function (Theorem \ref{msa}) are presented. The  proof of Theorem \ref{msa} is given in Section 3. The initialization for the MSA argument is presented in Section 4.  The whole MSA argument for the Green's function is presented in Section 5. The proof of localization for the $N$-particle system \eqref{n}(Theorem \ref{mainthm}) is completed in Section 6. Some useful lemmas are included in the Appendix.

	\section{Some Properties Of $n$-Particle Cubes}
	\subsection{Notations}\
	
	For any vector $\textbf x=(x_1,\cdots,x_n)\in\Z^{nd}$, define the projection mapping as
	\begin{align*}
		\Pi_{j}:\ \textbf x=(x_1,\cdots,x_n)\mapsto x_j,\ \ \ j=1,\cdots,n.
	\end{align*}
	The support $\Pi \textbf{x}\subset \R^d$ of a configuration $\textbf{x}$ and the full project $\Pi\Lambda$ of a subset $\Lambda\subset \Z^{nd}$ are defined by 
	\begin{align*} 
		\Pi \textbf x=\bigcup_{j=1}^n\{\Pi_j \textbf x\}\subset \Z^d,\ \ \ \  \Pi\Lambda=\bigcup_{j=1}^n \Pi_j \Lambda \subset \Z^d.
	\end{align*}
	A particle projection $\Pi_J\Lambda$, determined by a given nonempty subset $J\subseteq\{1,\cdots,n\}$ is defined by 
	\begin{align*}
		\Pi_J\Lambda=\bigcup_{j\in J}\Pi_j\Lambda \subset \Z^d.
	\end{align*}
	
	Here and below, $\Lambda^{(n)}_L(\textbf{u})$ stands for the $n$-particle lattice cube of size $2L$ around $\textbf{u}=(u_1,\cdots, u_n)$, where $u_j=\left(u^{(1)}_j,\cdots,u^{(d)}_j\right)\in\Z^d$:
	\begin{align}\label{twocube}
		\begin{split}
			\Lambda^{(n)}_L(\textbf{u})&=\left( \times_{j=1}^n \Pi_j\Lambda^{(n)}_L(\textbf u) \right)=\left( \times_{j=1}^n \Lambda^{(1)}_L(u_j) \right)\\
			&=\left( \times_{j=1}^n \times_{i=1}^d [u_j^{(i)}-L,u_j^{(i)}+L] \right)\cap \Z^{nd}.
		\end{split}
	\end{align}
	
	The cardinality of cube $\Lambda^{(n)}_L(\textbf{u})$ is denoted by $\big|\Lambda^{(n)}_L(\textbf{u})\big|$  and the cardinality of cube $\varPi_j\Lambda^{(n)}_L(\textbf{u})$ by $\big|\varPi_j\Lambda^{(n)}_L(\textbf{u})\big|$. $\textbf H^{(n)}_{\Lambda^{(n)}_L(\textbf{u})}$ is a Hermitian operator in the Hilbert space $\ell^2\big(\Lambda^{(n)}_L(\textbf{u})\big)$ of dimension $\big|\Lambda^{(n)}_L(\textbf{u})\big|$.

	\subsection{Separability of Cubes}
	\begin{defn}\label{spear}\

		\begin{itemize}
			\item[(i)]A pair of $n$-particle cubes $\Lambda_L^{(n)}(\textbf{u})$ and $\Lambda_L^{(n)}(\textbf{v})$ is called \textit{weakly separable} if there exists a non-empty subset $J\subseteq\{1,\cdots,n\}$ with $J^C=\{1,\cdots,n\}\setminus J$ such that either
			\begin{align*}
				\Pi_J\Lambda_{L+\mathrm{r}_0}^{(n)}(\textbf{u})\cap\left( 	\Pi_{J^C}\Lambda_{L+\mathrm{r}_0}^{(n)}(\textbf{u})\cup 	\Pi\Lambda_{L+\mathrm{r}_0}^{(n)}(\textbf{v})\right)=\emptyset
			\end{align*}	
			or
			\begin{align*}
				\Pi_J\Lambda_{L+\mathrm{r}_0}^{(n)}(\textbf{v})\cap\left( 	\Pi_{J^C}\Lambda_{L+\mathrm{r}_0}^{(n)}\textbf{v})\cup 	\Pi\Lambda_{L+\mathrm{r}_0}^{(n)}(\textbf{u})\right)=\emptyset.
			\end{align*}	
			\item[(ii)] A pair of $n$-particle cubes $\Lambda_L^{(n)}(\textbf{u})$ and $\Lambda_L^{(n)}(\textbf{v})$ is called \textit{separable} if and only if it is weakly separable and $\|\textbf{u}-\textbf{v}\|> 11nL$.
			\item[(iii)] A pair of $n$-particle cubes $\Lambda_L^{(n)}(\textbf{u})$ and $\Lambda_L^{(n)}(\textbf{v})$ is called \textit{completely separable}  if 
			\begin{align*}
				\Pi\Lambda_{L+\mathrm{r}_0}^{(n)}(\textbf{u})\cap	\Pi\Lambda_{L+\mathrm{r}_0}^{(n)}(\textbf{v})=\emptyset.
			\end{align*}	
		\end{itemize}
	\end{defn}
	\begin{lem}\label{weaks}
		Let $n\geq2$, $\textbf{u},\textbf{v}\in\Z^{nd}$. The following statements hold
		\begin{itemize}
			\item[(a)] Given an $n$-particle cube $\Lambda_L^{(n)}(\textbf{u})$, there exists a finite collection of $n$-particle cubes $\big\{\Lambda_{A(n,L)}^{(n)}(\textbf{u}^{i})\big\}$ of radius $A(n,L)\leq 2n(L+\mathrm{r}_0)$, where $i$ runs from $1$ to a number $K(\textbf{u},n,L)$ with $K(\textbf{u},n,L)\leq n^n$, such that if 
			\begin{align*}
				v\notin	\bigcup_{i=1}^{K(\textbf{u},n,L)}\Lambda^{(n)}_{A(n,L)}(\textbf{u}^i),
			\end{align*}
			then the pair of cubes $\Lambda_L^{(n)}(\textbf{u})$ and $\Lambda_L^{(n)}(\textbf{v})$ is weakly separable.
			\item[(b)] If $\|\textbf{u}\|>\|\textbf{v}\|+2(L+\mathrm{r}_0)$, the cubes  $\Lambda_L^{(n)}(\textbf{u})$ and $\Lambda_L^{(n)}(\textbf{v})$ are weakly separable.
		\end{itemize}	
	\end{lem}
	
	\begin{proof}
		The proof of (a) and (b) can be seen in Lemma 3.3.2 of \cite{CS2014}.
	\end{proof}
	\begin{rem}
		From the assertion (b) of Lemma \ref{weaks}, when $L\geq r_0$,  for any  $\textbf u\in\Z^{nd}$ with $\|\textbf u\|\geq 11nL$, the  pair of cubes  $\Lambda_L^{(n)}(\textbf{u})$ and $\Lambda_L^{(n)}(\textbf{0})$ is separable.
	\end{rem}
	
	\begin{defn}\label{pifi}
		Consider the following subset in $\Z^{nd}$ which is called the diagonal:
		\begin{align*}
			\mathbb{D}_{0}=\{\textbf{x}=(x,\cdots,x),\ x\in\Z^d\}.
		\end{align*}	
		An $n$-particle cube $\Lambda^{(n)}_{L}(\textbf{u})$ is called \textit{fully interactive} (FI) if 
		\begin{align*}
			\text{dist}(\textbf{u},\mathbb{D}_0)=\min _{x\in\Z^d}\max_{j:1\leq j\leq n}|u_j-x|\leq 2n (L+\mathrm{r}_0),
		\end{align*}
		and \textit{partially interactive} (PI) otherwise, 
		where $\mathrm{r}_0$ is the range of the interaction $U$ defined in (\ref{u2}). 
	\end{defn}
	
	\begin{lem}[Lemma 4.2.2 in \cite{CS2014}]
		If an $n$-particle cube $\Lambda_L^{(n)}(\textbf{u})$ is PI, then there exists two complementary non-empty index subset $J,J^C\subset \{1,\cdots,n\}$ such that 
		\begin{align*}
			\Pi_J\Lambda_{L+r_0}^{(n)}(\textbf{u})\bigcap \Pi_{J^C}\Lambda_{L+r_0}^{(n)}(\textbf{u})=\emptyset.
		\end{align*}
	\end{lem}
	
	\begin{lem}[Lemma 4.2.3 in \cite{CS2014}]\label{FIS}
		Let $n\geq 2$. If the $n$-particle cubes $\Lambda_L^{(n)}(\textbf{u})$ and $\Lambda_L^{(n)}(\textbf{v})$ are FI, and $\|\textbf{u}-\textbf{v}\|>n(10L+8r_0)$, then $\Lambda_L^{(n)}(\textbf{u})$ and $\Lambda_L^{(n)}(\textbf{v})$  are completely separable.
		
	\end{lem}	
	
	\subsection{Singular and Non-singular $n$-particle Cubes}\
	
	The proof of Theorem \ref{mainthm}  mainly uses MSA technique in its $n$-particle version. Most of the time we will work with finite-volume version $\mathbf{H}^{(n)}_{\Lambda^{(n)}_{L}(\textbf{u})}(\omega)$ of the operator $\mathbf{H}^{(n)}(\omega)$, acting in the space $\ell^2\big(\Lambda^{(n)}_{L}(\textbf{u})\big) \simeq \mathbb{C}^{\big|\Lambda^{(n)}_{L}(\textbf{u})\big|}$:
	\begin{align*} 
		\mathbf{H}^{(n)}_{\Lambda^{(n)}_{L}(\textbf{u})}(\omega)=\frac{1}{g}\left(\mathbf{T}_{\Lambda^{(n)}_{L}(\textbf{u})}+\mathbf{U}_{\Lambda^{(n)}_{L}(\textbf{u})}\right)+\mathbf{V}_{\Lambda^{(n)}_{L}(\textbf{u})}(\omega).
	\end{align*}	
	Here, $\mathbf{T}_{\Lambda^{(n)}_{L}}$ stands for the long-range hopping  in $\Lambda^{(n)}_{L}$ with Dirichlet boundary condition on $\partial^+\Lambda^{(n)}_{L}=\left\{\textbf{y}\in \Z^{nd}\setminus\Lambda^{(n)}_{L}: \ \text{dist}\big(y,\Lambda^{(n)}_{L}\big)=1 \right\}$. For any $\bm{\phi}\in\ell^2(\Lambda^{(n)}_{L})$, one has 
	\begin{align*}
		\Big(\mathbf{T}_{\Lambda^{(n)}_{L}}\bm{\phi}\Big) (\textbf x)=\sum_{\textbf y\in\Lambda^{(n)}_{L}} \mathbf{T}(\textbf x,\textbf y)\bm{\phi}(\textbf y),\ \ \ \textbf x\in\Lambda^{(n)}_{L}.
	\end{align*}
	$\mathbf{U}_{\Lambda^{(n)}_{L}}$ and $\mathbf{V}_{\Lambda^{(n)}_{L}}$ stand for the restrictions of the multiplication operators $\mathbf{U}$ and $\mathbf{V}(\omega)$ to $\Lambda^{(n)}_{L}$.
	
	Fix $1\leq n \leq N$. For any vector $\textbf{x}=(x_1,\cdots,x_n)\in\Z^{nd}$,
	let $\langle \textbf{x}\rangle=\max\{1,\|\textbf{x}\|\}$ and $s_0^{(n)}  >\frac{nd}{2}$. 
	Define for $\bm{\phi}=\bm{\phi}(\textbf{x})\in\mathbb{C}^{\Z^{nd}}$
	and $s>0$ the Sobolev norm
	\begin{align}\label{s-norm}
		\|\bm{\phi}\|^2_s=C_0(s^{(n)}_0)\sum_{\textbf x\in\Z^{nd}}|\bm{\phi}(\textbf{x})|^2\langle \textbf{x}\rangle^{2s},
	\end{align}
	where $C_0(s^{(n)}_0)>0$ is fixed so that (for $s\geq s^{(n)}_0$)
	\begin{align*}
		\|\bm{\phi}_1\bm{\phi}_2\|_s\leq \frac{1}{2}\|\bm{\phi}_1\|_{s_0}\|\bm{\phi}_2\|_s+C(s)\|\bm{\phi}_1\|_s\|\bm{\phi}_2\|_{s_0},
	\end{align*}	
	with $C(s)>0$, $C(s^{(n)}_0)=1/2$ and 
	\begin{align*}
		(\bm{\phi}_1\bm{\phi}_2)(\textbf{x})=\sum_{\textbf{y}\in\Z^{nd}}\bm{\phi}_1(\textbf{x}-\textbf{y})\bm{\phi}_2(\textbf{y}).
	\end{align*}
	
	Let $X,Y\subset \Z^{nd}$ be finite sets. Define 	
	\begin{align*}
		\textbf{M}^{X}_{Y}=\{\mathcal{M}=\mathcal{M}(\textbf{x},\textbf{y})\in\C\}_{\textbf{x}\in X,\textbf{y}\in Y}
	\end{align*}	
	to be the  set of all complex matrices with column indexes in $X$ and row indexes in $Y$. If $X_1\subset X$ and $Y_1\subset Y$, we write $\mathcal{M}^{X_1}_{Y_1}=(\mathcal{M}(\textbf{x},\textbf{y}))_{\textbf{x}\in X_1,\textbf{y}\in Y_1}$ for any $\mathcal{M}\in\textbf{M}^{X}_{Y}$.
	
	\begin{defn}\label{sobolev}
		Let $\mathcal{M}\in\textbf{M}^X_Y$.	 Define  for $s\geq s_0$ the \textit{Sobolev norm }of $\mathcal{M}$ as:
		\begin{equation*}
			\|\mathcal{M}\|^2_s=C_0(s^{(n)}_0)\sum_{\textbf{v} \in X-Y}\Big(\sup_{\textbf{x}-\textbf{y}=\textbf{v}}|\mathcal{M}(\textbf{x},\textbf{y})|\Big)^2\langle \textbf{v}\rangle^{2s},
		\end{equation*}
		where $C_0(s^{(n)}_0)$ is defined in (\ref{s-norm}).  
	\end{defn}
	
	We have the following uesful properties of Sobolev norm for matrices (see \cite{Shi2021} for details).
	\begin{itemize}
		\item \textbf{(Interpolation property)}: Let $B, C, D$ be finite subsets of $\mathbb{Z}^{nd}$ and let  $\mathcal{P}_1 \in \textbf{M}^{C}_{D}, \ \mathcal{P}_2 \in \textbf{M}^{B}_{C}$. Then for $s \geq s^{(n)}_0$,
		\begin{equation*}
			\|\mathcal{P}_1\mathcal{P}_2\|_s \leq \frac{1}{2}  \|\mathcal{P}_1\|_{s^{(n)}_0}\|\mathcal{P}_2\|_s+\frac{C(s)}{2}\|\mathcal{P}_1\|_{s_0}\|\mathcal{P}_2\|_{s^{(n)}_0},
		\end{equation*}
		and 
		\begin{equation*}
			\begin{split}
				\|\mathcal{P}_1 \mathcal{P}_2\|_{s^{(n)}_0} &\leq \|\mathcal{P}_1\|_{s^{(n)}_0}\|\mathcal{P}_2\|_{s^{(n)}_0},\\
				\|\mathcal{P}_1 \mathcal{P}_2\|_{s} &\leq C(s)\|\mathcal{P}_1\|_{s}\|\mathcal{P}_2\|_{s},
			\end{split}
		\end{equation*}
		where $C(s)\geq 1$.
		
		\item \textbf{(Perturbation argument)}: If $\mathcal{D} \in \textbf{M}^{B}_{C}$ has a left inverse $\mathcal{N} \in \textbf{M}^{C}_{B}$ (i.e., $\mathcal{ND=I}$, where $\mathcal{I}$ the identity matrix), then for all $\mathcal{P} \in  \textbf{M}^{B}_{C}$ with $$\|\mathcal{P}\|_{s^{(n)}_0}\|\mathcal{D}\|_{s^{(n)}_0} \leq \frac{1}{2}, $$ the matrix $\mathcal{D+P}$ has a left inverse $\mathcal{N}_{\mathcal{P}}$ that satisfies 
		
		\begin{align}\label{pa1}
			\|\mathcal{N}_{\mathcal{P}}\|_{s^{(n)}_0} \leq 2\|\mathcal{N}\|_{s^{(n)}_0},
		\end{align}
		\begin{align}\label{pa2}
			&\|\mathcal{N}_{\mathcal{P}}\|_{s} \leq C(s)\big(\|\mathcal{N}\|_s+\|\mathcal{N}\|^2_{s^{(n)}_0}\|\mathcal{P}\|_s\big), \ \text{for} \  s\geq s^{(n)}_{0}.
		\end{align}
		Moreover, if $\|\mathcal{P}\|\cdot \|\mathcal{N}\| \leq \frac{1}{2}$, then
		\begin{equation*}
			\|\mathcal{N}_{\mathcal{P}}\| \leq 2\|\mathcal{N}\|. 
		\end{equation*}
	\end{itemize}

	For $\Lambda^{(n)}_L(\textbf{u})\subseteq \Z^{nd}$, define the Green's function (if it exists) 
	$$\textbf G^{(n)}_{\Lambda^{(n)}_L(\textbf{u})}(E)=\Big(\textbf H^{(n)}_{\Lambda^{(n)}_L(\textbf{u})}-E\Big)^{-1}, \quad E\in\R.$$
	\begin{defn}
		Fix $1\leq n\leq N$. Let $\tau_n\geq 0$, $\delta \in (0,1)$ and $$\frac{nd}{2}<s^{(n)}_0\leq s\leq r_n <r-\frac{nd}{2},$$  we call $\Lambda^{(n)}_L(\textbf{u})$ is $(E,\delta)$-\textit{non-singular} ($(E,\delta)$-NS for short) if $\textbf G^{(n)}_{\Lambda^{(n)}_L(\textbf{u})}(E)$ exists and satisfies 
		\begin{align*}
			\Big\|\textbf G^{(n)}_{\Lambda^{(n)}_L(\textbf{u})}(E)\Big\|_s\leq L^{\tau_n+\delta s} \ \ \ \text{for } \forall \ s\in[s^{(n)}_0,  {r}_n].
		\end{align*}
		Otherwise, we call	$\Lambda^{(n)}_L(\textbf{u})$ is $(E,\delta)$-\textit{singular} ($(E,\delta)$-S).  
	\end{defn}	
	\begin{rem}\label{decreas}
		Let $\zeta\in(\delta,1)$ and $\tau_n+\delta r_n< \zeta r_n$. Suppose that $\Lambda^{(n)}_L(\textbf{u})$ is $(E,\delta)$-NS. Then we have for $L\geq \tilde L_0(\zeta,\tau_n,\delta,r_n,n,d)$ and $\|\textbf x'-\textbf x''\|>L/2$ that
		\begin{equation*}
			\Big|\textbf G^{(n)}_{\Lambda_{L}(\textbf{u})}(E)(\textbf{x}',\textbf{x}'')\Big| \leq \|\textbf{x}'-\textbf{x}''\|^{-(1-\zeta)r_n}. 
		\end{equation*}
	\end{rem}
	\begin{defn}
		Given $E\in\R$, $\textbf{u}\in \Z^{nd}$ and $L>1$, we call the $n$-particle cube $\Lambda^{(n)}_L(\textbf{u})$ is $E$-\textit{resonant} ($E$-R) if the spectrum of  $\textbf H^{(n)}_{\Lambda^{(n)}_L(\textbf{u})}$ satisfies
		\begin{align*}
			\text{dist}\Big[E,\sigma\Big(\textbf H^{(n)}_{\Lambda^{(n)}_L(\textbf{u})}\Big)\Big]<L^{-\beta},
		\end{align*}	
		and $E$-\textit{non-resonant} ($E$-NR) otherwise.
	\end{defn}
	
	\subsection{The MSA Result.}\

	\begin{itemize}
		\item  Fix $N\geq 2$ and let $1\leq n\leq N$.
		\item Let interval $I=[-MN-1,MN+1]$.
		\item  Let  $\delta=\frac{1}{2}$.
		\item Let $p^{(n)}=18^{N-n}p_0$, $\quad$  $p_0 \geq 20Nd$.
		\item  Let $L_k=L^{4^{k}}_{0},\quad k\geq 1$.
		\item Let  $\beta\geq \frac{18^Np_0}{2\rho}$. 
		\item Let $\tau_n>4\beta+7s^{(n)}_{0}$ and $r_n > 2\tau_n+4s^{(n)}_0$.
		\item Let $\tau_n >\tau_{n-1}+\frac{3}{2}d$ and  $r_n< r_{n-1}-\frac{d}{2}$.

	\end{itemize}
	We introduce the property ${(\bm{SS}.N.I.k.n)}$:
	\begin{equation}\label{ss.k}
		(\bm{SS}.N.I.k.{n})\left.\begin{array}{cc}
			\forall \text{ pair of separable $n$-particle cubes}\ \Lambda^{(n)}_{L_k}(\textbf{u})\ \text{and}\ \Lambda^{(n)}_{L_k}(\textbf{v}): \\
			\mathbb{P}\Big\{\exists E\in I,\ \text{s.t. both}\ \Lambda^{(n)}_{L_k}(\textbf{u})\ \text{and}\ \Lambda^{(n)}_{L_k}(\textbf{v})\  \text{are}\ (E,\frac{1}{2})\text{-S}\Big\}<L_k^{-2p^{(n)}}.
		\end{array}\right.
	\end{equation}
	Now we give the MSA result on Green’s functions estimate.
	\begin{thm}\label{msa}
		There exist $L_0^{*}=L_0^{*}(n,d,\kappa,\rho,\mathrm{r}_0,M,\beta,\tau_n,r_n,s_0^{(n)})\in(0,+\infty)$  and $\tilde{g}^*=\tilde{g}^*\big(\kappa,\rho,n,d,p^{(n)},\tau_n,M,\beta,\mathrm{r}_0, r_n,s_0^{(n)}\big)\in(0,+\infty)$ such that the following statement holds: Assume that $L_0\geq L_0^{*}$ and $|g|\geq \tilde g^{*}$. Suppose that $(\bm{SS}.N.I.k.\tilde{n})$, $(\bm{SS}.N.I.k-1.\tilde{n})$ for all $\tilde{n}\in[1,n-1]$ and $(\bm{SS}.N.I.k.n)$ are fulfilled. Then the property $(\bm{SS}.N.I.k+1.n)$ also holds true.  
	\end{thm}	
	The procedure of deducing property  $(\bm{SS}.N.I.k+1.n)$ from  $(\bm{SS}.N.I.k.\tilde{n})$, $(\bm{SS}.N.I.k-1.\tilde{n})$ and $(\bm{SS}.N.I.k.n)$ is done here separately for the following three cases:
	\begin{itemize}
		\item[(I)]  Both $\Lambda^{(n)}_{L_{k+1}}(\textbf{u})$ and 	$\Lambda^{(n)}_{L_{k+1}}(\textbf{v})$ are PI (a PI pair).
		\item[(II)]  Both $\Lambda^{(n)}_{L_{k+1}}(\textbf{u})$ and 	$\Lambda^{(n)}_{L_{k+1}}(\textbf{v})$ are FI (a FI pair).
		\item[(III)]  One of the cubes is PI, while the other is FI (a mixed pair). 
	\end{itemize}	
	Case (I)-(III) are treated in the next sections.

	\section{The Induction of Iteration Lemmas}\label{induction}
	
	\subsection{Analysis of Partially Interactive Pairs}\
	
	The following statement gives a hint of how the PI property of a cube $\Lambda^{(n)}_{L_{k+1}}(\textbf{u})$  will be used in the course of the induction step $k\rightsquigarrow k+1$.
	
	\begin{lem}\label{nrnt}
		Fix an energy $E\in\R$. Consider an $n$-particle PI cube with canonical decomposition $\Lambda^{(n)}_{L_{k}}(\textbf{u})=\Lambda^{(n')}_{L_{k}}(\textbf{u}')\times \Lambda^{(n'')}_{L_{k}}(\textbf{u}''),\ n=n'+n''$. Assume that
		\begin{itemize}
			\item[(a)] $\Lambda^{(n)}_{L_{k}}(\textbf{u})$ is $E$-NR.
			\item[(b)] $\forall$ eigenvalue $E_{j}^{''}\in\sigma\Big(\textbf H^{(n'')}_{\Lambda^{(n'')}_{L_{k}}(\textbf{u}'')}\Big)$, $\Lambda^{(n')}_{L_{k}}(\textbf{u}')$ is $\big(E-E_{j}^{''}, \frac{1}{2}\big)$-NS.
			\item[(c)] $\forall$ eigenvalue $E_{i}^{'}\in\sigma\Big(\textbf H^{(n')}_{\Lambda^{(n')}_{L_{k}}(\textbf u')}\Big)$, $\Lambda^{(n'')}_{L_{k}}(\textbf{u}'')$ is $\big(E-E_{i}^{'}, \frac{1}{2}\big)$-NS.
		\end{itemize}	
		Then $\Lambda^{(n)}_{L_{k}}(\textbf{u})$ is $\big(E,\frac{1}{2}\big)$-NS.
	\end{lem}
	\begin{proof}
		Since $\Lambda^{(n)}_{L_{k}}(\textbf{u})$ is PI,  $\textbf H^{(n)}_{\Lambda^{(n)}_{L_k}(\textbf{u})}$ admits the decomposition
		\begin{align*}
			\textbf H^{(n)}_{\Lambda^{(n)}_{L_k}(\textbf{u})}=\textbf H^{(n')}_{\Lambda^{(n')}_{L_k}(\textbf{u}')}\otimes \textbf{1}''+\textbf{1}'\otimes \textbf H^{(n'')}_{\Lambda^{(n'')}_{L_k}(\textbf{u}'')}.
		\end{align*}
		Thus its eigenvalues are the sums $E_{i,j}=E_i^{'}+E_j^{''}$ where $E_{i}^{'}\in\sigma\Big(\textbf H^{(n')}_{\Lambda^{(n')}_{L_{k}}(\textbf{u}')}\Big)$ and $E_{j}^{''}\in\sigma\Big(\textbf H^{(n'')}_{\Lambda^{(n'')}_{L_{k}}(\textbf{u}'')}\Big)$. Eigenfunctions of $\textbf H^{(n)}_{\Lambda^{(n)}_{L_k}(\textbf{u})}$ can be chosen in the form $\Psi_{i,j}=\bm \phi_i\otimes\bm \psi_j$ where $\{\bm \phi_i\}_{i\in\Z^{n'd}}$ are eigenfunctions of $\textbf H^{(n')}_{\Lambda^{(n')}_{L_k}(\textbf{u}')}$ and $\{\bm \psi_j\}_{j\in\Z^{n''d}}$ are eigenfunctions of $\textbf H^{(n'')}_{\Lambda^{(n'')}_{L_k}(\textbf{u}'')}$. For each pair $(E_i^{'},E_j^{''})$, the non-resonance assumption $|E-(E_i^{'}+E_j^{''})|\geq L_k^{-\beta}$ reads as  $|(E-E_i^{'})-E_j^{''}|\geq L_k^{-\beta}$ and $|(E-E_j^{''})-E_i^{'}|\geq L_k^{-\beta}$, which means $\Lambda^{(n')}_{L_{k}}(\textbf{u}')$ is $(E-E_{j}^{''})$-NR and $\Lambda^{(n'')}_{L_{k}}(\textbf{u}'')$ is $(E-E_{i}^{'})$-NR.
		
		Therefore we can write
		\begin{equation*} 
			\begin{split}
				\textbf G^{(n)}_{\Lambda^{(n)}_{L_{k}}(\textbf{u})}(\textbf{x},\textbf{y},E)&=\sum_{E_i^{'}}\sum_{E_j^{''}}\frac{\bm \phi_i(\textbf{x}')\bm \psi_j(\textbf{x}'')\bm \phi_i(\textbf{y}')\bm \psi_j(\textbf{y}'')}{(E_i^{'}+E_j^{''})-E}\\
				&=\sum_{E_i^{'}} P_1(\textbf{x}',\textbf{y}')G^{(n'')}_{\Lambda^{(n'')}_{L_{k}}(\textbf{u}'')}(\textbf{x}'',\textbf{y}'',E-E_i^{'})\\
				&=\sum_{E_j^{''}} P_2(\textbf{x}'',\textbf{y}'')G^{(n')}_{\Lambda^{(n')}_{L_{k}}(\textbf{u}')}(\textbf{x}',\textbf{y}',E-E_j^{''}).
			\end{split}
		\end{equation*}
		
		Let 
		$$\textbf k\in \Lambda^{(n)}_{L_{k}}(\textbf{u})-\Lambda^{(n)}_{L_{k}}(\textbf{u}) $$ 
		with  
		$$\textbf k'\in \Lambda^{(n')}_{L_{k}}(\textbf{u}')-\Lambda^{(n')}_{L_{k}}(\textbf{u}'),\quad \textbf k''\in \Lambda^{(n'')}_{L_{k}}(\textbf{u}'')-\Lambda^{(n'')}_{L_{k}}(\textbf{u}'').$$ 
		For any ${s}^{(n)}_0\leq s<{r}_n<r-\frac{nd}{2}$, we have
		\begin{align*}
			\Big\|\textbf G^{(n)}_{\Lambda^{(n)}_{L_{k}}(\textbf{u})}(E)\Big\|^2_s
			=&\ C_0(s^{(n)}_0)\sum_{\textbf k}	\Big(\sup_{\textbf{x}-\textbf{y}=\textbf{k}}\Big|\textbf G^{(n)}_{\Lambda^{(n)}_{L_{k}}(\textbf{u})}(\textbf{x},\textbf{y},E)\Big| \Big)^2 \langle\textbf{k}\rangle^{2s}\\
			=&\ C_0(s^{(n)}_0)\Big(\sum_{\textbf{k}:\ \|\textbf k''\|<\|\textbf k'\|}	\Big(\sup_{\textbf{x}-\textbf{y}=\textbf{k}}\Big|\textbf G^{(n)}_{\Lambda^{(n)}_{L_{k}}(\textbf{u})}(\textbf{x},\textbf{y},E)\Big|\Big)^2 \langle \textbf k'\rangle^{2s}\\
			&\     \quad +\sum_{\textbf{k}:\ \|\textbf k''\|\geq\|\textbf k'\|}	\Big(\sup_{\textbf{x}-\textbf{y}=\textbf{k}}\Big|\textbf G^{(n)}_{\Lambda^{(n)}_{L_{k}}(\textbf{u})}(\textbf{x},\textbf{y},E)\Big|\Big)^2 \langle \textbf k''\rangle^{2s}\Big)\\
			\leq&\  C_0(s^{(n)}_0)3^{n''d}\big|\Lambda^{(n'')}_{L_{k}}(\textbf u'')\big|^2\cdot \max_{E''_j}\\
			&\ \quad  \quad   \quad   \ \sum_{\textbf k'}\Big(\sup_{\textbf x'-\textbf y'=\textbf k'}\Big|\textbf G^{(n')}_{\Lambda^{(n')}_{L_{k}}(\textbf u')}(\textbf x',\textbf y',E-E_j^{''})\Big|\Big)^2 \langle  \textbf k'\rangle^{2s+n''d}\\
			&\ \quad  + C_0(s^{(n)}_0)3^{n'd}\big|\Lambda^{(n')}_{L_{k}}(\textbf u')\big|^2 \cdot \max_{E'_i}\\
			&\ \quad  \quad  \quad  \ \sum_{\textbf k''}\Big(\sup_{\textbf x''-\textbf y''=\textbf k''}\Big|\textbf G^{(n'')}_{\Lambda^{(n'')}_{L_{k}}(\textbf u'')}(\textbf x'',\textbf y'',E-E_i^{'})\Big|\Big)^2 \langle  \textbf k''\rangle^{2s+n'd}\\
			\leq &\ 3^{n''d} (2L_k+1)^{2n''d}\max_{E''_j}\Big\|\textbf G^{(n')}_{\Lambda^{(n')}_{L_{k}}(\textbf u')}(E-E_j^{''})\Big\|_{s+\frac{n''d}{2}}^2\\
			&\ \quad    +3^{n'd} (2L_k+1)^{2n'd}\max_{E'_i}\Big\|\textbf G^{(n'')}_{\Lambda^{(n'')}_{L_{k}}(\textbf u'')}(E-E_i^{'})\Big\|_{s+\frac{n'd}{2}}^2\\
			\leq&  3^{n''d} (2L_k+1)^{2n''d}L_k^{2(\tau_{n'}+\frac{1}{2}(s+\frac{n''d}{2}))}\\
			& \ \quad +3^{n'd} (2L_k+1)^{2n'd}L_k^{2(\tau_{n''}+\frac{1}{2}(s+\frac{n'd}{2}))}.
		\end{align*}
		The last inequality holds because  $$s+\frac{n''d}{2}<r_n+\frac{n''d}{2}<r_{n'}<r-\frac{n'}{2}d,$$ $$\quad s+\frac{n'd}{2}<r_n+\frac{n'd}{2}<r_{n''}<r-\frac{n''}{2}d.$$ Furthermore, since $\tau_n > \tau_m+\frac{3}{2}(n-m)d$, one has
		\begin{align*}
			\Big\|\textbf G^{(n)}_{\Lambda^{(n)}_{L_{k}}(\textbf{u})}(E)\Big\|^2_s
			&\leq \frac{1}{2}L_k^{2(\tau_{n'}+\frac{s}{2}+\frac{5+\epsilon}{4}n''d)}+\frac{1}{2}L_k^{2(\tau_{n''}+\frac{s}{2}+\frac{5+\epsilon}{4}n'd)}\\
			&\leq  L_k^{2(\tau_n+\frac{1}{2}s)},
		\end{align*}
		which implies 
		\begin{align*}
			\Big\|\textbf G^{(n)}_{\Lambda^{(n)}_{L_{k}}(\textbf{u})}(E)\Big\|_s\leq L_k^{ \tau_n+\frac{1}{2} s},
		\end{align*}	
		i.e., $\Lambda^{(n)}_{L_{k}}(\textbf{u})$ is $\big(E,\frac{1}{2}\big)$-NS.
	\end{proof}
	
	\begin{defn}\
		
		\begin{itemize}
			\item[(1)] We say that an $n$-particle cube 
			$\Lambda^{(n)}_{L_{k+1}}(\textbf u)$ is \textit{$(\frac{1}{2},I)$-tunneling} ($(\frac{1}{2},I)$-T) if there exist $E\in I$ and two separable $\big(E,\frac{1}{2}\big)$-S cubes  $\Lambda^{(n)}_{L_{k}}(\textbf x), \Lambda^{(n)}_{L_{k}}(\textbf y)\subset  \Lambda^{(n)}_{L_{k+1}}(\textbf u)$.
			
			\item[(2)] An  $n$-particle PI cube $\Lambda^{(n)}_{L_{k+1}}(\textbf{u})$ with the canonical decomposition
			$$\Lambda^{(n)}_{L_{k+1}}(\textbf{u})=\Lambda^{(n')}_{L_{k+1}}(\textbf u')\times \Lambda^{(n'')}_{L_{k+1}}(\textbf u'')$$
			is called \textit{$(\frac{1}{2},I)$-partially tunneling} ($(\frac{1}{2},I)$-PT) if  at least one of the cubes $\Lambda^{(n')}_{L_{k+1}}(\textbf u'), \Lambda^{(n'')}_{L_{k+1}}(\textbf u'')$ is \textit{$(\frac{1}{2},I)$-tunneling}. Otherwise, the cube $\Lambda^{(n)}_{L_{k+1}}(\textbf{u})$  is called \textit{$(\frac{1}{2},I)$-non-partially tunneling} ($(\frac{1}{2},I)$-NPT).
		\end{itemize}
	\end{defn}

	\begin{lem}\label{p-pt}
		Assume that the property $(\bm{SS}.N.I.k.\tilde{n})$ holds true for all $\tilde{n}\in[1,n-1]$. Then $\forall$ $n$-particle $\mathrm{PI}$ cube $\Lambda^{(n)}_{L_{k+1}}(\textbf{u})$,
		\begin{equation*}
			\mathbb{P}\Big\{ \Lambda^{(n)}_{L_{k+1}}(\textbf{u}) \text{ is } \big(\frac{1}{2}, I\big)\text{-PT} \Big\}\leq 3^{2(n-1)d}L_{k+1}^{-\frac{p^{(n-1)}}{2}+2(n-1)d}.
		\end{equation*}	
	\end{lem}
	\begin{proof}
		The PI cube $\Lambda^{(n)}_{L_{k+1}}(\textbf{u})$  with the canonical decomposition
		\begin{equation*}
			\Lambda^{(n)}_{L_{k+1}}(\textbf{u}) =\Lambda^{(n')}_{L_{k+1}}(\textbf u')\times\Lambda^{(n'')}_{L_{k+1}}(\textbf u'')
		\end{equation*}	
		is $(\frac{1}{2},I)$-PT  if at least one of the following events occurs:
		\begin{equation*}
			\mathcal{IT}_1:=\Big\{ \Lambda^{(n')}_{L_{k+1}}(\textbf u')\  \text{is}\  \big(\frac{1}{2},I\big)\text{-T} \Big\}
		\end{equation*}	
		or
		\begin{equation*}
			\mathcal{IT}_2:=\Big\{ \Lambda^{(n'')}_{L_{k+1}}(\textbf u'')\  \text{is}\  \big(\frac{1}{2},I\big)\text{-T}\Big\}.
		\end{equation*}	
		Therefore,
		\begin{equation}\label{pt}
			\mathbb{P}\Big\{ \Lambda^{(n)}_{L_{k+1}}(\textbf{u}) \text{ is } \big(\frac{1}{2}, I\big)\text{-PT} \Big\}\leq 2\max\Big\{\mathbb{P}\{\mathcal{IT}_1\},\mathbb{P}\{\mathcal{IT}_2\}\Big\}.
		\end{equation}	
		We will now focus on the estimate of  $\mathbb{P}\{\mathcal{IT}_1\}$, the probability $\mathbb{P}\{\mathcal{IT}_2\}$ is bounded in the same way.
		
		The cube $\Lambda^{(n')}_{L_{k+1}}(\textbf u')$ is $(\frac{1}{2},I)$-T iff for some $E\in I$ it contains a pair of separable $(E,\frac{1}{2})$-S cubes $\Lambda^{(n')}_{L_{k}}(\textbf x')$ and $\Lambda^{(n')}_{L_k}(\textbf y')$. The number of such pairs is bounded by
		\begin{align*}
			\frac{1}{2}\big|\Lambda^{(n')}_{L_{k+1}}(\textbf u')\big|^2&=\frac{1}{2}(2L_{k+1}+1)^{2n'd}\\
			&\leq \frac{3^{2(n-1)d}}{2}L_{k+1}^{2(n-1)d}.
		\end{align*}	
		Next,  since the property $(\bm{SS}.N.I.k.\tilde{n})$ holds, one has for any  pair of separable cubes $\Lambda^{(n')}_{L_k}(\textbf x')$ and $\Lambda^{(n')}_{L_k}(\textbf y')$ that
		\begin{align*}
			\mathbb{P}&\Big\{\exists E\in I \ \text{s.t.\ both}\  \Lambda^{(n')}_{L_k}(\textbf x')\  \text{and} \ \Lambda^{(n')}_{L_k}(\textbf y') \ \text{are}\  \big(E,\frac{1}{2}\big)\text{-S} \Big\}\\
			&\leq L_k^{-2p^{(n')}}  \leq L_{k+1}^{-\frac{p^{(n-1)}}{2}}.
		\end{align*}	
		Therefore,
		\begin{equation*}
			\mathbb{P}\{\mathcal{IT}_1\}\leq \frac{3^{2(n-1)d}}{2}L_{k+1}^{-\frac{p^{(n-1)}}{2}+2(n-1)d}.
		\end{equation*}	
		Similarly,
		\begin{equation*}
			\mathbb{P}\{\mathcal{IT}_2\}\leq \frac{3^{2(n-1)d}}{2}L_{k+1}^{-\frac{p^{(n-1)}}{2}+2(n-1)d}.
		\end{equation*}	
		Taking into account (\ref{pt}), the assertion of the lemma follows.
	\end{proof}	
	
	\begin{lem}\label{ns}
		Fix the integer $n\in\{2,\cdots,N\}$. For a given integer $k\geq 0$, consider an $n$-particle $\mathrm{PI}$ cube $\Lambda^{(n)}_{L_{k+1}}(\textbf{u})$ with a canonical decomposition $\Lambda^{(n)}_{L_{k+1}}(\textbf{u}) =\Lambda^{(n')}_{L_{k+1}}(\textbf u')\times\Lambda^{(n'')}_{L_{k+1}}(\textbf u'')$ and
		\begin{equation*}
			\text{dist}\Big(\Pi\Lambda^{(n')}_{L_{k+1}}(\textbf u'),\ \Pi\Lambda^{(n'')}_{L_{k+1}}(\textbf u'')\Big)>\mathrm{r}_0.
		\end{equation*}	
		Assume that
		\begin{itemize}
			\item[(a)] $\Lambda^{(n)}_{L_{k+1}}(\textbf{u})$ is $(\frac{1}{2},I)$-NPT,
			\item[(b)] $\Lambda^{(n)}_{L_{k+1}}(\textbf{u})$ is $E$-NR for some $E\in I$.
		\end{itemize}
		Then $\exists \hat L_{0}^{*}:= \hat L_0^{*}(n,d,\mathrm{r}_0,M,\beta,\tau_n,r_n,s_0^{(n)})\in(0,+\infty)$ such that if $L_0> \hat L_0^*$, the cube $\Lambda^{(n)}_{L_{k+1}}(\textbf{u})$ is $\big(E, \frac{1}{2}\big)$-NS.
	\end{lem}	
	\begin{proof}
		Since $\Lambda^{(n)}_{L_{k+1}}(\textbf{u})$ is $\mathrm{PI}$,  $\textbf H^{(n)}_{\Lambda^{(n)}_{L_{k+1}}(\textbf{u})}$ admits the decomposition
		\begin{align*}
			\textbf H^{(n)}_{\Lambda^{(n)}_{L_{k+1}}(\textbf{u})}=\textbf H^{(n')}_{\Lambda^{(n')}_{L_{k+1}}(\textbf u')}\otimes \textbf{1}''+\textbf{1}' \otimes\textbf H^{(n'')}_{\Lambda^{(n'')}_{L_{k+1}}(\textbf u'')}.
		\end{align*}
		Thus its eigenvalues are the sums $E_{i,j}=E_i^{'}+E_j^{''}$ where $E_{i}^{'}\in\sigma\Big(\textbf H^{(n')}_{\Lambda^{(n')}_{L_{k+1}}(\textbf u')}\Big), i=1,\cdots, \big|\Lambda^{(n')}_{L_{k+1}}(\textbf u')\big|$, and $E_{j}^{''}\in\sigma\Big(\textbf H^{(n'')}_{\Lambda^{(n'')}_{L_{k+1}}(\textbf u'')}\Big), j=1,\cdots,\big|\Lambda^{(n'')}_{L_{k+1}}(\textbf u'')\big|$.  
		Eigenfunctions of $\textbf H^{(n)}_{\Lambda^{(n)}_{L_{k+1}}(\textbf{u})}$ can be chosen in the form $\bm \Psi_{i,j}=\bm \phi_i\otimes\bm \psi_j$, where $\{\bm \phi_i\}_{i\in\Z^{n'd}}$ are eigenfunctions of $\textbf H^{(n')}_{\Lambda^{(n')}_{L_{k+1}}(\textbf u')}$ and $\{\bm \psi_j\}_{j\in\Z^{n''d}}$ are eigenfunctions of $\textbf H^{(n'')}_{\Lambda^{(n'')}_{L_k}(\textbf u'')}$. 
		
		Further,  $\Lambda^{(n)}_{L_{k+1}}(\textbf{u})$ is $E$-NR. Therefore, for all $E_{i}^{'}\in\sigma\Big(\textbf H^{(n')}_{\Lambda^{(n')}_{L_{k+1}}(\textbf u')}\Big)$, $\Lambda^{(n'')}_{L_{k+1}}(\textbf u'')$ is $(E-E_{i}^{'})$-NR. In addition, by the assumption of $(\frac{1}{2},I)$-NPT, \textbf{if} $E-E'_i\in I$, the cube $\Lambda^{(n'')}_{L_{k+1}}(\textbf u'')$ cannot contain two separable $\big(E-E_i^{'}, \frac{1}{2}\big)$-S cubes of size $L_k$.
		
		As a result, if
		\begin{itemize}
			\item There exists a cube  $\Lambda^{(n'')}_{L_{k}}(\textbf x)\subset\Lambda^{(n'')}_{L_{k+1}}(\textbf u'')$ which is $(E-E_i^{'},\frac{1}{2})$-S and in addition,
			\item No  pair of separable cubes of radius $L_k$ lying in $\Lambda^{(n'')}_{L_{k+1}}(\textbf u'')$ is $(E-E_i^{'},\frac{1}{2})$-S,
		\end{itemize}	
		then from Lemma \ref{weaks},  all $\big(E-E_i^{'},\frac{1}{2}\big)$-S cubes of radius $L_k$  lying in $\Lambda^{(n'')}_{L_{k+1}}(\textbf u'')$ can be covered by at most $(n'')^{n''}+1$ cubes  of radius $11n''L_{k}$. Since $\Lambda^{(n'')}_{L_{k+1}}(\textbf u'')$ is $(E-E_i^{'})$-NR, we can apply the coupling Lemma \ref{coupling} and Lemma \ref{SC}. If the relation (\ref{condition}) holds true for $n$ replaced by $n''$, then   $\Lambda^{(n'')}_{L_{k+1}}(\textbf u'')$ is  $(E-E_i^{'},\frac{1}{2})$-NS. 
		
		\textbf{ If $E-E'_i\notin I$}, then 
		$$\min_{\textbf x\in\Lambda^{(n'')}_{L_{k+1}}(\textbf u'')}|E-E'_i-\textbf V(\textbf x;\omega)|\geq M.$$
		Let $$\textbf D_{\Lambda^{(n'')}_{L_{k+1}}(\textbf u'')}=\text{diag}_{\textbf x\in\Lambda^{(n'')}_{L_{k+1}}(\textbf u'')}(E-E'_i-\textbf V(\textbf x;\omega)),$$ then $$\Big\|\textbf D^{-1}_{\Lambda^{(n'')}_{L_{k+1}}(\textbf u'')}\Big\|_s\leq \frac{C_0(s)}{M}.$$ Also,  from the  \textbf{perturbation argument} \eqref{pa1}-\eqref{pa2}, if $g$ is sufficiently large, one has 
		\begin{align*}
			\Big\|\textbf G^{(n'')}_{\Lambda^{(n'')}_{L_{k+1}}(\textbf u'')}(E-E'_i)\big\|_s=&\Big\|\Big(-\textbf D_{\Lambda^{(n'')}_{L_{k+1}}(\textbf u'')}+\frac{1}{g}\Big(\textbf{T}_{\Lambda^{(n'')}_{L_{k+1}}(\textbf u'')}+\textbf{U}_{\Lambda^{(n'')}_{L_{k+1}}(\textbf u'')}\Big)\Big)^{-1}\Big\|_s \\
			\leq & C(s)\Big(\frac{1}{M}+\frac{1}{g}\frac{1}{M^2}\Big)\leq L_{k+1}^{\tau_{n''}},
		\end{align*}
		i.e., $\Lambda^{(n'')}_{L_{k+1}}(\textbf u'')$ is  $(E-E_i^{'},\frac{1}{2})$-NS.

		Similarly, we conclude that the cube $\Lambda^{(n')}_{L_{k+1}}(\textbf u')$ is $(E-E_j^{''},\frac{1}{2})$-NS.
		Finally, from Lemma \ref{nrnt}, we see that the cube $\Lambda^{(n)}_{L_{k+1}}(\textbf{u})$ is $(E, \frac{1}{2})$-NS.
	\end{proof}

	\begin{lem}\label{msapi} Assume that $(\bm{SS},N,I,k,\tilde{n})$ is fulfilled for all $\tilde{n} \in [1,n-1]$.
		Let $\Lambda^{(n)}_{L_{k+1}}(\textbf{u}),\  \Lambda^{(n)}_{L_{k+1}}(\textbf{v})$ be separable $\mathrm{PI}$ cubes. Then 
		\begin{align*}
			\mathbb{P}&\Big\{\exists E\in I, \ {\rm s.t.\ both}\  \Lambda^{(n)}_{L_{k+1}}(\textbf{u})\  {\rm and} \ \Lambda^{(n)}_{L_{k+1}}(\textbf{v}) \ {\rm are}\  \big(E,\frac{1}{2}\big)\text{-}{\rm S} \Big\}\\
			&\leq \tilde{C}(n,d,\rho,\kappa)L_{k+1}^{-8p^{(n)}},
		\end{align*}	
		provided that $L_0$ larger enough.
	\end{lem}
	\begin{proof}
		Consider the following events
		\begin{align*}
			\mathcal{S}&=\Big\{\exists E\in I, \ \text{s.t. both}\  \Lambda^{(n)}_{L_{k+1}}(\textbf{u})\  \text{and} \ \Lambda^{(n)}_{L_{k+1}}(\textbf{v}) \ \text{are}\  \big(E,\frac{1}{2}\big)\text{-S}\Big\},\\
			\mathcal{PT}(\textbf{u})&=\Big\{\ \Lambda^{(n)}_{L_{k+1}}(\textbf{u}) \text{ is}\  (\frac{1}{2},I)\text{-PT} \Big\},\\
			\mathcal{PT}(\textbf{v})&=\Big\{\ \Lambda^{(n)}_{L_{k+1}}(\textbf{v}) \text{ is}\  (\frac{1}{2},I)\text{-PT} \Big\},\\
			\mathcal{R}&=\Big\{\exists E\in I, \ \text{s.t. both}\  \Lambda^{(n)}_{L_{k+1}}(\textbf{u})\  \text{and} \ \Lambda^{(n)}_{L_{k+1}}(\textbf{v}) \ \text{are}\  E\text{-R} \Big\}.
		\end{align*}	
		By Lemma \ref{ns}, if the cube $\Lambda^{(n)}_{L_{k+1}}(\textbf{u})$ is $(\frac{1}{2},I)$-NPT and $E$-NR for some $E\in I$, then for $L_0$ sufficiently large,  $\Lambda^{(n)}_{L_{k+1}}(\textbf{u})$ is $(E,\frac{1}{2})$-NS. The same is true, of course, for $\Lambda^{(n)}_{L_{k+1}}(\textbf{v})$. Therefore, $\mathcal{S}\subset\mathcal{PT}(\textbf{u})\cup\mathcal{PT}(\textbf{v})\cup\mathcal{R}$. Further, by virtue of Lemma \ref{p-pt}, we have that
		\begin{equation*}
			\max\big\{\mathbb{P}\{\mathcal{PT}(\textbf{u})\},\mathbb{P}\{\mathcal{PT}(\textbf{v})\}\big\}\leq 3^{2(n-1)d}L_{k+1}^{-\frac{p^{(n-1)}}{2}+2(n-1)d}.
		\end{equation*}	
		Next, by Theorem \ref{wegnern},
		\begin{align*}
			\mathbb{P}\{\mathcal{R}\}\leq n \kappa^{-1}4^{\rho} 3^{(2n+1)d}L_{k+1}^{-\beta\rho+(2n+1)d}.
		\end{align*}	
		By summing up, we finally obtain
		\begin{align}\label{pmsapi}
			\nonumber	\mathbb{P}\{\mathcal{S}\}&\leq \mathbb{P}\{\mathcal{PT}(\textbf{u})\}+\mathbb{P}\{\mathcal{PT}(\textbf{v})\}+\mathbb{P}\{\mathcal{R}\}\\
			\nonumber	&\leq 2\cdot3^{2(n-1)d}L_{k+1}^{-\frac{p^{(n-1)}}{2}+2(n-1)d}+ n \kappa^{-1}4^{\rho} 3^{(2n+1)d}L_{k+1}^{-\beta\rho+(2n+1)d}\\
			&\leq \tilde{C}(n,d,\rho,\kappa)\big(L_{k+1}^{-\frac{p^{(n-1)}}{2}+2(n-1)d}+L_{k+1}^{-\beta\rho+(2n+1)d}\big)\\
			\nonumber &\leq \tilde{C}(n,d,\rho,\kappa)L_{k+1}^{-8p^{(n)}}.
		\end{align}
		The last inequality holds because
		\begin{equation}
			\begin{split}
				-\frac{p^{(n-1)}}{2}+(2n-1)d&=-8p^{(n)}+(2n-1)d-p^{(n)} < -8p^{(n)},\\
			\end{split}
		\end{equation}
		and
		\begin{equation}
			\begin{split}
				-\beta \rho +(2n+1)d&= -\frac{18^{N}p_0}{2}+(2n+1)d\\
				&=-9\cdot 18^{N-1}p_0+(2n+1)d\\
				&<-8p^{(n)}.
			\end{split}
		\end{equation}
	\end{proof}	
	
	Let $K_{PI}\left(\Lambda^{(n)}_{L_{k+1}}(\textbf{u}),E\right)$ denote the largest cardinality of a collection of pairwise separable, partially interactive $(E,\frac{1}{2})$-S cubes of radius $L_k$ contained in the cube $\Lambda^{(n)}_{L_{k+1}}(\textbf{u})$. Further, set
	\begin{align*}
		K_{PI}\big(\Lambda^{(n)}_{L_{k+1}}(\textbf{u})\big)=\sup _{E\in I}K_{PI}\big(\Lambda^{(n)}_{L_{k+1}}(\textbf{u}),E\big).
	\end{align*}	
	Then Lemma \ref{msapi} leads directly to the following corollary.
	
	\begin{cor}\label{kpi}
		Suppose $L_0$ large enough and assume that the property $(\bm{SS},N,I,k-1,\tilde{n})$ is fulfilled for all $\tilde{n} \in [1,n-1]$, then for any $n$-particle cube $\Lambda^{(n)}_{L_{k+1}}(\textbf{u})$,
		\begin{align*}
			\mathbb{P}\Big \{K_{PI}\big(\Lambda^{(n)}_{L_{k+1}}(\textbf{u})\big)\geq 2  \Big\}\leq\tilde{C}_1(n,d,\rho,\kappa) L_{k+1}^{-\frac{17}{8}p^{(n)}}.
		\end{align*}
	\end{cor}
	
	\begin{proof}
		From Lemma \ref{msapi}, for any separable pair of PI cubes  $\Lambda^{(n)}_{L_{k}}(\textbf{x}), \Lambda^{(n)}_{L_{k}}(\textbf{y})$ inside $\Lambda^{(n)}_{L_{k+1}}(\textbf{u})$, the estimate of the form (\ref{pmsapi}) holds true:
		\begin{align*}
			\mathbb{P}&\Big\{\exists E\in I, \ \text{s.t.\ both}\  \Lambda^{(n)}_{L_{k}}(\textbf{x})\  \text{and} \ \Lambda^{(n)}_{L_{k}}(\textbf{y}) \ \text{are}\  \big(E,\frac{1}{2}\big)\text{-S} \Big\}\\
			&\leq \tilde{C}(n,d,\rho,\kappa)\big(L_{k}^{-\frac{p^{(n-1)}}{2}+2(n-1)d}+L_{k}^{-\beta\rho+(2n+1)d}\big).
		\end{align*}	
		The number of pairs of cubes of radius $L_k$ inside $\Lambda^{(n)}_{L_{k+1}}(\textbf{u})$ is bounded by
		\begin{align*}
			\frac{1}{2}(2L_{k+1}+1)^{2nd}\leq \frac{3^{2nd}}{2}L_{k+1}^{2nd}.
		\end{align*}
		So we finally get the inequality
		\begin{align*}
			\mathbb{P}\Big \{K_{PI}\big(\Lambda^{(n)}_{L_{k+1}}(\textbf{u})\big)\geq 2  \Big\}
			&\leq  \tilde{C}_1(n,d,\rho,\kappa)\big(L_{k+1}^{-\frac{p^{(n-1)}}{8}+\frac{(5n-1)d}{2}}+L_{k+1}^{-\frac{\beta\rho}{4}+\frac{(10n+1)d}{4}}\big)\\
			&\leq\tilde{C}_1(n,d,\rho,\kappa) L_{k+1}^{-\frac{17}{8}p^{(n)}},
		\end{align*}
		since 
		\begin{equation}
			\begin{split}
				-\beta \rho+(10n+1)d &\leq  \frac{-17\cdot 18^{N-1}p_0-18^{N-1}p_0+2(10n+1)d}{2}\\
				& < \frac{17p^{(n)}}{2}.
			\end{split}
		\end{equation}
	\end{proof}

	\subsection{Analysis of Fully Interactive Pairs}
	
	\begin{lem}\label{separable}
		Let $L_0>8\mathrm{r}_0$. If $n$-particle cubes  $\Lambda^{(n)}_{L_k}(\textbf{u}), \Lambda^{(n)}_{L_k}(\textbf{v})$, $k\geq0$ are $\mathrm{FI}$ and $\|\textbf{u}-\textbf{v}\|\geq 11nL_k$, then these cubes are  completely separable. Moreover, the random operators $\textbf H^{(n)}_{\Lambda^{(n)}_{L_k}(\textbf{u})}$ and $\textbf H^{(n)}_{\Lambda^{(n)}_{L_k}(\textbf{v})}$ are independent (as the values of random potential $V({x},\omega)$ are IID).
	\end{lem}
	\begin{proof}The proof can be deduced from Lemma \ref{FIS}.
	\end{proof}

	Let $K_{FI}\big(\Lambda^{(n)}_{L_{k+1}}(\textbf{u}),E\big)$ denote the largest cardinality of a collection of pairwise separable, FI cubes of  radius $L_k$ contained in the cube $\Lambda^{(n)}_{L_{k+1}}(\textbf{u})$. Further, set
	\begin{align*}
		K_{FI}\big(\Lambda^{(n)}_{L_{k+1}}(\textbf{u})\big)=\sup _{E\in I}K_{FI}\big(\Lambda^{(n)}_{L_{k+1}}(\textbf{u}),E\big).
	\end{align*}	
	
	\begin{lem}\label{kfi}
		Suppose $L_0$ large enough and assume that the property $(\bm{SS},N,I,k,n)$ holds true,
		then for any $n$-particle cube $\Lambda^{(n)}_{L_{k+1}}(\textbf{u})$,
		\begin{align*}
			\mathbb{P}\left \{K_{FI}\big(\Lambda^{(n)}_{L_{k+1}}(\textbf{u})\big)\geq 12 \right\}
			\leq C'(n,d)L^{-\frac{7}{3}p^{(n)}}_{k+1}.
		\end{align*}
	\end{lem}	
	\begin{proof}
		Suppose $\exists$ FI cubes $\Lambda^{(n)}_{L_k}(\textbf{x}^{(1)}),\cdots,\Lambda^{(n)}_{L_k}(\textbf{x}^{(12)})\subset \Lambda^{(n)}_{L_{k+1}}(\textbf{u})$ such that any two of them are separable. From Lemma \ref{separable}, $\forall$ pair
		$\Lambda^{(n)}_{L_k}(\textbf{x}^{(2i-1)}), \Lambda^{(n)}_{L_k}(\textbf{x}^{(2i)})$,  the respective (random) operators $\textbf H^{(n)}_{\Lambda^{(n)}_{L_k}(\textbf{x}^{(2i-1)})}$, $\textbf H^{(n)}_{\Lambda^{(n)}_{L_k}(\textbf{x}^{(2i)})}$ are independent, and so are their spectra and Green's functions. Moreover, the pairs of operators 
		\begin{align}\label{hpair}
			\Big(\textbf H^{(n)}_{\Lambda^{(n)}_{L_k}(\textbf{x}^{(2i-1)})},\textbf H^{(n)}_{\Lambda^{(n)}_{L_k}(\textbf{x}^{(2i)})} \Big),\ \ i=1,\cdots,6
		\end{align}
		form an independent family. Now for $i=1,\cdots, 6$, set 
		\begin{align*}
			\mathcal{A}_i=\Big\{\exists E\in I, \ \text{s.t. both}\  \Lambda^{(n)}_{L_{k}}(\textbf{x}^{(2i-1)})\  \text{and} \ \Lambda^{(n)}_{L_{k}}(\textbf{x}^{(2i)}) \ \text{are}\  \big(E,\frac{1}{2}\big)\text{-S} \Big\}.
		\end{align*}	
		Then from the property  $(\bm{SS},N,I,k,n)$ , $\mathbb{P}(\mathcal{A}_i)\leq L_k^{-2p^{(n)}}$ and by virtue of independence of the events $\mathcal{A}_1,\cdots,\mathcal{A}_{6}$, we obtain 
		\begin{align*}\mathbb{P}\Big(\bigcap_{i=1}^{6}\mathcal{A}_i\Big)\leq L_k^{-12p^{(n)}}.\end{align*}
		To complete the proof, note that the number of all possible collections of $12$ cubes of radius $L_k$ inside $\Lambda^{(n)}_{L_{k+1}}(\textbf{u})$  including fully interactive ones, is bounded by $(2L_{k+1}+1)^{12nd}/12!$. Dividing $12$ cubes into $6$ pairs and applying the inductive assumption of the lemma, we obtain
		\begin{align*}
			\mathbb{P}\Big \{K_{FI}\big(\Lambda^{(n)}_{L_{k+1}}(\textbf{u})\big) \geq 12 \Big\}&\leq \frac{(2L_{k+1}+1)^{12nd}}{12!}L_{k+1}^{-\frac{12p^{(n)}}{4}}\\ 
			&\leq \frac{3^{12nd}}{12!}L_{k+1}^{-\frac{12p^{(n)}}{4}+12nd}\\
			& \leq C'(n,d)L^{-\frac{7}{3}p^{(n)}}_{k+1}.
		\end{align*}
	\end{proof}

	\begin{lem}\label{FIs+1}
		Assume that the property $(\bm{SS},N,I,k-1,\tilde{n})$ holds true for all $\tilde{n} \in [1,n-1]$ and the property $(\bm{SS},N,I,k,n)$ holds true. Then for any pair of separable $n$-particle $\mathrm{FI}$ cubes $\Lambda^{(n)}_{L_{k+1}}(\textbf{u})$, $\Lambda^{(n)}_{L_{k+1}}(\textbf{v})$, the following bound holds true:
		\begin{align}\label{msak+1fi}
			\mathbb{P}\Big\{\exists E\in I, \ {\rm s.t.\ both}\  \Lambda^{(n)}_{L_{k+1}}(\textbf{u})\  {\rm and} \ \Lambda^{(n)}_{L_{k+1}}(\textbf{v}) \ {\rm are}\  \big(E,\frac{1}{2}\big)\text{-}{\rm S} \Big\}\leq \frac{1}{3} L_{k+1}^{-2p^{(n)}}
		\end{align}	
		provided that $L_0$ large enough.
	\end{lem}
	
	\begin{proof}
		Consider the following events
		\begin{align*}
			\mathcal{S}&=\Big\{\exists E\in I, \ \text{s.t. both}\  \Lambda^{(n)}_{L_{k+1}}(\textbf{u})\  \text{and} \ \Lambda^{(n)}_{L_{k+1}}(\textbf{v}) \ \text{are}\  \big(E,\frac{1}{2}\big)\text{-S}\Big\},\\
			\mathcal{FI}&=\Big\{ \max \Big[K_{FI}\big(\Lambda^{(n)}_{L_{k+1}}(\textbf{u})\big), K_{FI}\big(\Lambda^{(n)}_{L_{k+1}}(\textbf{v})\big) \Big] \geq 12 \Big\},\\
			\mathcal{PI}&=\left\{\max \Big[K_{PI}\big(\Lambda^{(n)}_{L_{k+1}}(\textbf{u})\big), K_{PI}\big(\Lambda^{(n)}_{L_{k+1}}(\textbf{v})\big) \Big] \geq 2 \right\},\\
			\mathcal{R}&=\Big\{\exists E\in I, \ \text{s.t. both}\  \Lambda^{(n)}_{L_{k+1}}(\textbf{u})\  \text{and} \ \Lambda^{(n)}_{L_{k+1}}(\textbf{v}) \ \text{are}\  E\text{-R} \Big\}.
		\end{align*}	
		
		Firstly, we show that $\mathcal{S} \subset \mathcal{FI} \cup \mathcal{PI} \cup \mathcal{R}$. Indeed, if $\omega \notin \mathcal{FI} \cup \mathcal{PI} \cup \mathcal{R}$, then for any $E\in I$, neither of the cubes $\Lambda^{(n)}_{L_{k+1}}(\textbf{u}), \Lambda^{(n)}_{L_{k+1}}(\textbf{v})$ can contain 
		\begin{itemize}
			\item Twelve or more separable $\mathrm{FI} \  (E,\frac{1}{2})$-S cubes of radius $L_k$.
			
			\item  Two or more separable $\mathrm{PI} \  (E,\frac{1}{2})$-S cubes of radius $L_k$.
		\end{itemize}
		
		Therefore, all  $\mathrm{FI}$ $(E,\frac{1}{2})$-S cubes of radius $L_k$ in  $\Lambda^{(n)}_{L_{k+1}}(\textbf{u})$ can be covered by at most $11$ cubes of radius $11nL_k$. From Lemma \ref{weaks}, all  $\mathrm{PI}$ $(E,\frac{1}{2})$-S cubes of radius $L_k$ in  $\Lambda^{(n)}_{L_{k+1}}(\textbf{u})$ can be covered by at most $n^n+1$ cubes  of radius $11nL_k$. 
		
		The same is true, of course, for  $\Lambda^{(n)}_{L_{k+1}}(\textbf{v})$.
		
		Also, since $\omega \notin \mathcal{R}$, at least one of the cubes $\Lambda^{(n)}_{L_{k+1}}(\textbf{u}), \Lambda^{(n)}_{L_{k+1}}(\textbf{v})$ is $E$-NR. From the coupling Lemma \ref{coupling} and Lemma \ref{SC}, we see that either  $\Lambda^{(n)}_{L_{k+1}}(\textbf{u})$ or $ \Lambda^{(n)}_{L_{k+1}}(\textbf{v})$ is $(E,\frac{1}{2})$-NS. Hence, the assumption $\omega \notin  \mathcal{FI} \cup \mathcal{PI} \cup \mathcal{R}$ implies  $\omega \notin \mathcal{S} $. 
		
		Finally, from
		Corollary \ref{kpi}, Lemma \ref{kfi} and Theorem \ref{wegnern}, we conclude that, for $L_0$ large enough
		\begin{equation*}
			\begin{split}
				\mathbb{P}\{\mathcal{S}\} &\leq \mathbb{P}\{\mathcal{FI}\}+\mathbb{P}\{\mathcal{PI}\}+\mathbb{P}\{\mathcal{R}\}\\
				&\leq \tilde{C}_1(n,d,\rho,\kappa)L^{-\frac{17}{8}p^{(n)}}_{k+1}+C'(n,d)L^{-\frac{7}{3}p^{(n)}}_{k+1}
				+\tilde{C}(n,d,\rho,\kappa)L^{-\beta \rho+(2n+1)d}_{k+1}\\
				&\leq \frac{1}{3}L^{-2p^{(n)}}_{k+1}.
			\end{split}
		\end{equation*}
	\end{proof}

	\subsection{Analysis of Mixed Pairs}\
	
	The analysis of a mixed pair of cubes $\Lambda^{(n)}_{L_{k+1}}(\textbf{u})$, $\Lambda^{(n)}_{L_{k+1}}(\textbf{v})$ requires a combination of methods used in the analysis of PI and FI pairs. The main  outcome of this section is the following lemma.
	\begin{lem}\label{mixmsa}
		Assume that the following properties are fulfilled:
		\begin{itemize}
			\item $(\bm{SS},N,I,k,\tilde{n})$ for all $\tilde{n} \in [1,n-1]$.
			\item $(\bm{SS},N,I,k-1,\tilde{n})$ for all $\tilde{n} \in [1,n-1]$.
			\item $(\bm{SS},N,I,k,n)$.
		\end{itemize}
		Let $\Lambda^{(n)}_{L_{k+1}}(\textbf {u}), \Lambda^{(n)}_{L_{k+1}}(\textbf {v})$ be separable $n$-particle cubes, where $\Lambda^{(n)}_{L_{k+1}}(\textbf{u})$ is $\mathrm{FI}$ and $\Lambda^{(n)}_{L_{k+1}}(\textbf{v})$ is $\mathrm{PI}$. Then 
		\begin{align*}
			\mathbb{P}\Big\{\exists E\in I, \ {\rm s.t.\ both}\  \Lambda^{(n)}_{L_{k+1}}(\textbf{u})\  {\rm and} \ \Lambda^{(n)}_{L_{k+1}}(\textbf{v}) \ {\rm are}\  \big(E,\frac{1}{2}\big)\text{-}{\rm S} \Big\}\leq  \frac{1}{3}L_{k+1}^{-2p^{(n)}}
		\end{align*}	
		provided that $L_0$ large enough.
	\end{lem}
	
	\begin{proof} 
		Consider the following  events
		\begin{align*}
			\mathcal{S}&=\Big\{\exists E\in I, \ \text{s.t. both}\  \Lambda^{(n)}_{L_{k+1}}(\textbf{u})\  \text{and} \ \Lambda^{(n)}_{L_{k+1}}(\textbf{v}) \ \text{are}\  \big(E,\frac{1}{2}\big)\text{-S}\Big\},\\
			\mathcal{FI}(\textbf{u})&=\left\{\  K_{FI}\left(\Lambda^{(n)}_{L_{k+1}}(\textbf{u})\right)\geq 12 \right\},\\
			\mathcal{PI}(\textbf{u})&=\left\{\  K_{PI}\left(\Lambda^{(n)}_{L_{k+1}}(\textbf{u})\right)\geq 2 \right\},\\
			\mathcal{PT}(\textbf{v})&=\Big\{\ \Lambda^{(n)}_{L_{k+1}}(\textbf{v}) \text{ is}\  \big(\frac{1}{2},I\big)\text{-PT} \Big\},\\
			\mathcal{R}&=\Big\{\exists E\in I, \ \text{s.t. both}\  \Lambda^{(n)}_{L_{k+1}}(\textbf{u})\  \text{and} \ \Lambda^{(n)}_{L_{k+1}}(\textbf{v}) \ \text{are}\  E\text{-R} \Big\}.
		\end{align*}	
		Then 
		\begin{align*}
			\mathcal{S}&=[\mathcal{S}\cap\mathcal{PT}(\mathbf{v})]\cup[\mathcal{S}\cap(\mathcal{PT}(\mathbf{v}))^C]\\
			&\subset \mathcal{PT}(\mathbf{v})\cup[\mathcal{S}\cap(\mathcal{PT}(\mathbf{v}))^C]\\
			&= \mathcal{PT}(\mathbf{v})\cup[\mathcal{S}\cap(\mathcal{PT}(\mathbf{v}))^C\cap\mathcal{R}]\cup[\mathcal{S}\cap(\mathcal{PT}(\mathbf{v}))^C\cap\mathcal{R}^C]\\
			&\subset \mathcal{PT}(\mathbf{v})\cup\mathcal{R}\cup[\mathcal{S}\cap(\mathcal{PT}(\mathbf{v}))^C\cap\mathcal{R}^C].
		\end{align*}
		
		Next, we consider the event  $\mathcal{S}\cap(\mathcal{PT}(\mathbf{v}))^C\cap\mathcal{R}^C$. If $\omega\in\mathcal{S}\cap(\mathcal{PT}(\mathbf{v}))^C\cap\mathcal{R}^C$, then for some $E\in I$, one of the cubes $\Lambda^{(n)}_{L_{k+1}}(\textbf{u}), \Lambda^{(n)}_{L_{k+1}}(\textbf{v})$ is $E$-NR. Since $\Lambda^{(n)}_{L_{k+1}}(\textbf{v})$ is $(\frac{1}{2},I)$-NPT, if the cube $\Lambda^{(n)}_{L_{k+1}}(\textbf{v})$ is $E$-NR for some $E\in I$,  then  $\Lambda^{(n)}_{L_{k+1}}(\textbf{v})$ is $(E,\frac{1}{2})$-NS by Lemma \ref{ns}. This violates the definition of the event $S$. Therefore $\Lambda^{(n)}_{L_{k+1}}(\textbf{u})$ must be  $E$-NR and $(E,\frac{1}{2})$-S. Owing to Lemma \ref{FIs+1}, it must be 
		
		\begin{align*}
			&K_{FI}\left(\Lambda^{(n)}_{L_{k+1}}(\textbf{u})\right)\geq 12,\\
			&K_{PI}\left(\Lambda^{(n)}_{L_{k+1}}(\textbf{u})\right)\geq 2,
		\end{align*}
		for $L_0$ large enough,
		which means 
		\begin{align*}
			\mathcal{S}\cap(\mathcal{PT}(\mathbf{v}))^C\cap\mathcal{R}^C &\subseteq \Big\{\Lambda^{(n)}_{L_{k+1}}(\textbf{u}) \text{ is } E\text{-NR and}\  \big(E,\frac{1}{2}\big)\text{-S for some } E\in I\Big\}\\
			&\subseteq \mathcal{FI}(\textbf{u})\cup \mathcal{PI}(\textbf{u}).
		\end{align*}
		Finally, from Lemma \ref{msapi}, Corollary \ref{kpi}, Lemma \ref{kfi} and Theorem \ref{wegnern}, we can get 
		\begin{align*}
			\mathbb{P}\{\mathcal{S}\}&\leq\mathbb{P}\{\mathcal{PT}(\textbf{v})\}+\mathbb{P}\{\mathcal{R}\}+\mathbb{P}\{ \mathcal{S}\cap(\mathcal{PT}(\mathbf{v}))^C\cap\mathcal{R}^C\}\\
			&\leq \mathbb{P}\{\mathcal{PT}(\textbf{v})\}+\mathbb{P}\{\mathcal{R}\}+\mathbb{P}\{\mathcal{FI}(\textbf{u})\}+\mathbb{P}\{\mathcal{PI}(\textbf{u})\} \\
			&\leq \tilde{C}(N,d,\rho,\kappa)L^{-8p^{(n)}}_{k+1}
			+\tilde{C}(n,d,\rho,\kappa)L^{-\beta \rho+(2n+1)d}_{k+1}
			+C'(n,d)L^{-\frac{7}{3}p^{(n)}}_{k+1}\\
			&\quad\quad+\tilde{C}_1(N,d,\rho,\kappa)L^{-\frac{17}{8}p^{(n)}}_{k+1} \\
			&\leq \frac{1}{3} L^{-2p^{(n)}}_{k+1}
		\end{align*}
		for $L_0$ large enough.
	\end{proof}
	
	\section{The initialization of the MSA}\label{0-step}
	\begin{defn}
		Fix $N\geq 2$ and $1\leq n \leq N$. Given $E \in I$ and $\gamma>0$, an  $n$-particle cube  $\Lambda_L^{(n)}(\textbf{u})$ is called $(E, \gamma)$-\textit{Good}, if
		\begin{equation}
			\min_{\textbf x \in \Lambda^{(n)}_{L}(\textbf{u})}|\mathbf{V}(\textbf{x},\omega)-E|>\gamma.
		\end{equation}
		Otherwise, we call it $(E, \gamma)$-\textit{Bad}.
	\end{defn}
	\begin{lem}\label{NS-bad}
		Let $\gamma\geq \frac{1}{4} n^{-\frac{1}{\rho}}3^{-\frac{(2n+1)d}{\rho}}\kappa^{\frac{1}{\rho}}L_0^{-\frac{2p^{(n)}+(2n+1)d}{\rho}}$ and $\tau_n >\frac{2p^{(n)}+(2n+1)d}{\rho}$. There exists 
		$$\bar{L}_{0}^*= \bar{L}_{0}^*(n,d,\kappa,\rho,M,\mathrm{r}_0,s_0^{(n)},r_n),$$ 
		such that if $L_0\geq  \bar{L}_{0}^*$, there is some $\bar{g}_0^*=\bar{g}_0^*(n,d,M,\mathrm{r}_0,s^{(n)}_0,\gamma)$, so that for $|g|>\bar{g}_0^*$, if the cube $\Lambda^{(n)}_{L_0}(\textbf{u})$ is $(E,\gamma)$-Good, then it  is also  $(E,\frac{1}{2})$-NS.
	\end{lem}
	\begin{proof}
		Consider the truncated diagonal operator 
		$$\mathbf{D}_{\Lambda^{(n)}_{L_0}(\textbf{u})}=\mathbf{V}_{\Lambda^{(n)}_{L_0}(\textbf{u})}(\cdot,\omega)-E.$$
		From  Definition \ref{sobolev}, we have 
		\begin{equation}
			\Big\|\mathbf{D}^{-1}_{\Lambda^{(n)}_{L_0}}\Big\|_{s} \leq \frac{C(n,d,s)}{\gamma}, \quad s\geq s^{(n)}_0.
		\end{equation}
		Besides, we see that 
		\begin{equation}
			\Big\|g^{-1}\big(\mathbf{T}_{\Lambda^{(n)}_{L_0}(\textbf{u})}+\mathbf{U}_{\Lambda^{(n)}_{L_0}(\textbf{u})}\big)\mathbf{D}_{\Lambda^{(n)}_{L_0}(\textbf{u})}^{-1}\Big \|_{s^{(n)}_0} \leq C(n,d,M,s^{(n)}_0) g^{-1} \gamma^{-1} \leq \frac{1}{2},
		\end{equation}
		if $g\geq \bar g^*_0=2 C(n,d,M,r_0,s^{(n)}_0) \gamma^{-1}$. Then by the \textbf{perturbation argument} (\ref{pa1})-(\ref{pa2}) and 
		\begin{equation*}
			\mathbf{H}^{(n)}_{\Lambda^{(n)}_{L_0}(\textbf{u})}-E= g^{-1}\Big(\mathbf{T}_{\Lambda^{(n)}_{L_0}(\textbf{u})}+\mathbf{U}_{\Lambda^{(n)}_{L_0}(\textbf{u})}\Big)+\mathbf{D}_{\Lambda^{(n)}_{L_0}(\textbf{u})},
		\end{equation*}
		we have 
		\begin{equation*}
			\Big\|(\mathbf{H}^{(n)}_{\Lambda^{(n)}_{L_0}(\textbf{u})}-E)^{-1}\Big\|_{s^{(n)}_0} \leq 2 \Big\|\mathbf{D}_{\Lambda^{(n)}_{L_0}(\textbf{u})}^{-1}\Big\|_{s^{(n)}_0} \leq 2C(n,d,s)\gamma^{-1},
		\end{equation*}
		and 
		\begin{equation*}
			\begin{split}
				\Big\|(\mathbf{H}^{(n)}_{\Lambda^{(n)}_{L_0}(\textbf{u})}-E)^{-1}\Big\|_{s} &\leq C(s,n,d) (\gamma^{-1}+g^{-1}\gamma^{-2}) \\
				&\leq C(n,d,M,s,s^{(n)}_0)\gamma^{-1}.
			\end{split}
		\end{equation*}
		Let $s^{(n)}_0\leq s\leq r_n$. In order to show that the cube $\Lambda^{(n)}_{L_0}(\textbf{u})$ is $(E,\frac{1}{2})$-NS, it suffice to show that 
		\begin{equation*}
			C(n,d,M,r_0,r_n,s_0^{(n)})\gamma^{-1}\leq C(n,d,M,\kappa, \rho,s_0^{(n)},r_n)L_0^{\frac{2p^{(n)}+(2n+1)d}{\rho}} \leq L_0^{\tau_n}.
		\end{equation*}
		This can be verified by taking $\tau_n>\frac{2p^{(n)}+(2n+1)d}{\rho} $ and $L_0\geq \bar{L}_{0}^*(n,d,\kappa,\rho,M,s^{(n)}_0,r_n)$  large enough.
	\end{proof}
	
	\begin{thm}\label{PNS}
		Let $\mu$ be H\"older continuous of order $\rho>0$ (i.e., $\mathcal{K}_\rho(\mu)>0$). Fix $0<\kappa< \mathcal{K}_\rho(\mu)$ and $\tau_n>\frac{2p^{(n)}+(2n+1)d}{\rho}$. Then there exists 
		\begin{equation*}
			\tilde L^*_{0}=\tilde L^*_{0}\big(n,d,\kappa,\rho,M,\mathrm{r}_0,p^{(n)},\tau_n,r_n,s_0^{(n)}\big) 
		\end{equation*}
		such that if $L_0>L^*_{0}$, there is some $\tilde{g}_0^*=\tilde{g}_0^*\big(\kappa,\rho,n,d,p^{(n)},\tau_n,r_n,s_0^{(n)}\big)$, so that for $|g|>\tilde{g}_0^*$ and any pair of separable cubes  $ \Lambda^{(n)}_{L_0}(\textbf{u})$ and $ \Lambda^{(n)}_{L_0}(\textbf{v})$,
		\begin{align}\label{P0ns}
			\mathbb{P}\Big\{{\exists}\  E\in \mathbb{R},\ {\rm s.t. both}\ \Lambda^{(n)}_{L_0}(\textbf{u})\  {\rm and}\  \Lambda^{(n)}_{L_0}(\textbf{v})\ {\rm  are\ }  \big(E,\frac{1}{2}\big)\text{-{\rm{NS}}}\Big\}
			\leq L_0^{-2p^{(n)}}.
		\end{align}
	\end{thm}
	\begin{proof}
		From Lemma \ref{NS-bad} and Theorem \ref{wegnern}, for any pair  of separable  cubes  $ \Lambda^{(n)}_{L_0}(\textbf{u})$ and $ \Lambda^{(n)}_{L_0}(\textbf{v})$ , one has
		\begin{equation}\label{initial}
			\begin{split}
				\mathbb{P}&\left\{\exists \ E\in \mathbb{R}: {\rm\ both}\ \Lambda^{(n)}_{L_0}(\textbf{u})\  {\rm and}\  \Lambda^{(n)}_{L_0}(\textbf{v})\  {\rm are\ }  \big(E,\frac{1}{2}\big)\text{-S}\right\} \\
				&\leq   \mathbb{P}\left\{{\ \exists}\  E\in \mathbb{R}: \ {\rm \ both\ }\ \Lambda^{(n)}_{L_0}(\textbf{u})\  {\rm and}\  \Lambda^{(n)}_{L_0}(\textbf{v})\  {\rm are\ }  (E,\gamma)\text{-{\rm{Bad}}}\right\}\\
				& \leq \mathbb{P}\left\{\min_{\textbf{x}\in \Lambda^{(n)}_{L_0}(\textbf{u})}\ \min_{\textbf{y}\in \Lambda^{(n)}_{L_0}(\textbf{v})}|\mathbf{V}(\textbf{x},\omega)-\mathbf{V}(\textbf{y},\omega)|\leq 2\gamma \right\}\\
				&=\mathbb{P}\left\{\text{dist}\Big[\sigma\Big(\mathbf{V}^{(n)}_{\Lambda^{(n)}_{L_0}(\textbf{u})}\Big),\sigma\Big(\mathbf{V}^{(n)}_{\Lambda^{(n)}_{L_0}(\textbf{v})}\Big)\Big] \leq 2\gamma\right\}\\
				& \leq \ n 3^{(2n+1)d}L_0^{(2n+1)d} \kappa^{-1}(4\gamma)^{\rho}.
			\end{split}
		\end{equation}
		By choosing  $\gamma=\frac{1}{4} n^{-\frac{1}{\rho}}3^{-\frac{(2n+1)d}{\rho}}\kappa^{\frac{1}{\rho}}L_0^{-\frac{2p^{(n)}+(2n+1)d}{\rho}}$, the \textbf{RHS} of \eqref{initial} is bounded by $L_0^{-2p^{(n)}}$, which completes the proof.
	\end{proof}
	\begin{rem}
		Let $1\le n\leq N$, we can choose the largest $\tilde{L}_{0}^*$ and $\tilde{g}^*_0$ in Theorem \ref{PNS} such that the probability estimate holds for all $1\leq n\leq N$. Furthermore, let $L_1=L_0^4$, we can choose a bigger number  $\tilde{g}^{*}_1$ such that for $|g|>\tilde{g}^{*}_1$, the estimate (\ref{P0ns}) also holds for $n$-particle cubes of size $L_1$.
	\end{rem}
	
	\section{The conclusion of multi-particle multi-scale analysis}
	\begin{thm}
		Fix integers $N \geq 2 ,d\geq1, 1\leq n\leq N$ and  real numbers $\tau_n,p^{(n)},\beta,r_n$ as the assumptions of Theorem \ref{msa}. Then there exists 
		\begin{equation*}
			L^*_0= L^*_0\big(n,d,\kappa,\rho,M,\mathrm{r}_0,p^{(n)},\tau_n,r_n,s_0^{(n)}\big)\in \mathbb{N} 
		\end{equation*}
		such that for $L_0 \geq L^*_0$ and 
		$$|g| \geq g^*=g^*(N,d,\kappa, \rho,p_0,L_0,\beta) \in (0,+\infty),$$
		the property $(\bm{SS},N,I,k,n)$ holds true for all $k\geq 0$.
	\end{thm}
	
	\begin{proof}
		We will prove that for any   $1\leq n\leq N$, property  $(\bm{SS},N,I,k,n)$ holds true by induction.
		\begin{itemize} 
			\item $(\bm{SS},N,I,0,n)$ and  $(\bm{SS},N,I,1,n)$  are valid for any $1\leq n \leq N$  by using Theorem \ref{PNS}.
			\item $(\bm{SS},N,I,k,1)$ for all $k\geq 0$ can be easily verified  by simplifying  our proofs in Section \ref{induction}, because any pair of  disjoint  single particle cubes are completely separable. 
			\item Now, suppose that for any $1\leq \tilde{n}<n$,  $(\bm{SS},N,I,k,\tilde{n})$ holds  for all $k\geq 0$. Also, suppose that $(\bm{SS},N,I,\tilde{k},n)$ holds true for all $\tilde{k} \leq k$. Then, one can conclude that 
			
			$\textbf{1}:$  Property $(\bm{SS},N,I,k+1,n)$ holds true for all pairs of PI cubes using Lemma \ref{msapi}.
			
			$\textbf{2}:$  Property $(\bm{SS},N,I,k+1,n)$ holds true for all pairs of FI cubes using Lemma \ref{FIs+1}.
			
			$\textbf{3}:$  Property $(\bm{SS},N,I,k+1,n)$ holds true for all pairs of mixed cubes using Lemma \ref{mixmsa}.
		\end{itemize}
		Finally, we complete the proof.
	\end{proof}

	\section{the proof of the main theorem}
	Let $\bm{\phi}(\textbf{x}) \in \mathbb{C}^{\mathbb{Z}^{Nd}}$ satisfy $\mathbf{H}^{(N)}_{\omega}\bm{\phi}(\textbf{x})=E \bm{\phi}(\textbf x)$. Assume further that $\textbf G^{(N)}_{\Lambda}(E)$ exists  for some $\Lambda \subset \mathbb{Z}^{Nd}$. Then for any $x\in \Lambda$,  we have 
	\begin{equation}\label{possion}
		\bm{\phi}(\textbf{x})=-g^{-1}\sum_{\substack{\textbf{x}' \in \Lambda \\ \textbf{x}'' \notin \Lambda}}\textbf G^{(N)}_{\Lambda}(E)(\textbf{x},\textbf{x}')\mathbf{T}(\textbf{x}',\textbf{x}'')\bm{\phi}(\textbf{x}'').
	\end{equation}

	\begin{defn} \textbf{(Generalized Eigenfunction)} 
		Let $\epsilon_1>0$. An energy $E$ is called an $\epsilon_1$-generalized eigenvalue, if there exists some $\bm{\phi} \in \mathbb{C}^{\mathbb{Z}^{Nd}}$ satisfy 
		\begin{equation}\label{geeifu}
			\bm{\phi}(\textbf 0)=1, \quad |\bm\phi(\textbf n)| \leq C(1+|\textbf  n|)^{\frac{Nd}{2}+\epsilon_1} 
		\end{equation}
		and $\mathbf{H}^{(N)}_{\omega}\bm{\phi}=E\bm{\phi}$. We call such $\bm{\phi}$ the $\epsilon_1$ generalized eigenfunction.
	\end{defn}
	
	The Shnol’s Theorem of \cite{H2019} in long-range operator case is 
	\begin{lem}[\cite{H2019}]
		Let $r-2Nd>\epsilon$. Let $\Sigma_{\omega,\epsilon}$ be the set of all $\epsilon$-generalized eigenvalues of $\mathbf{H}^{(N)}_{\omega}$. Then  we have 
		\begin{equation*}
			\Sigma_{\omega,\epsilon} \subset \sigma(\mathbf{H}^{(N)}_{\omega}), \quad  v_{\omega}(\sigma(\mathbf{H}^{(N)}_{\omega}) \backslash \Sigma_{\omega,\epsilon}) =0, 
		\end{equation*}
		where $v_{\omega}$ denotes some complete spectral measure of $\mathbf{H}^{(N)}_{\omega}$.
	\end{lem}
	Based on the assumptions of lemmata in Section \ref{0-step} and coupling Lemma \ref{coupling}, we can choose some appropriate  parameters.
	
	Let $0<\epsilon\ll 1$ and
	\begin{equation*}
		s^{(n)}_0=\frac{nd}{2}+\epsilon, \quad \beta=\frac{18^N p_0}{2\rho}, \quad \tau_n=\frac{2\cdot 18^N p_0}{\rho}+\frac{7nd}{2}+O(\epsilon).   
	\end{equation*}
	From Remark \ref{decreas}, set $\zeta=\frac{19}{20}$, then for 
	$$r_N > \frac{20}{9}\tau_N,$$
	if $\Lambda_{L}(\textbf{x})$ is $(E,\frac{1}{2})$-NS, then 
	\begin{equation*}
		\Big\|G^{(N)}_{\Lambda_{L}(\textbf{x})}(E)\Big\| \leq L^{\frac{18^N p_0}{2\rho}}, 
	\end{equation*}
	\begin{equation*}
		\Big|G^{(N)}_{\Lambda_{L}(\textbf{x})}(E)(\textbf{x}'.\textbf{x}'')\Big| \leq \|\textbf{x}'-\textbf{x}''\|^{-\frac{r_{N}}{20}}, \quad \text{for }\|\textbf{x}'-\textbf{x}''\| >\frac{L}{2}. 
	\end{equation*}
	
	\begin{lem}\label{ngene}
		Let 
		$$ r_N>\frac{40 \cdot  18^N p_0}{9\rho}+ \frac{70Nd}{9}+O(\epsilon). $$
		If there  exist infinitely many $L_{k}$,  such  that $\Lambda_{L_k}(\textbf{0})$ is $(E,\frac{1}{2})$-NS, then $E$ is not a $\epsilon_1$-generalized eigenvalue of $\mathbf{H}^{(N)}_{\omega}$.
	\end{lem}
	\begin{proof}
		Suppose that  $E$ is a $\epsilon_1$-generalized eigenvalue of $\mathbf{H}^{(N)}_{\omega}$ and  $\bm{\phi}$ is a corresponding polynomial-bounded eigenfunctions.
		From the Possion's identity (\ref{possion}) and (\ref{T}),  (\ref{geeifu}), (\ref{sum}), $\forall \textbf x \in \Lambda^{(N)}_{L_{k}^{\frac{1}{2}}}(\textbf{0})$, one has
		\begin{align*}
			|\bm{\phi}(\textbf{x})| \leq& \sum_{\substack{ \textbf{x}' \in \Lambda^{(N))}_{L_k}(\textbf 0) \\  \textbf{x}'' \notin \Lambda^{(N)}_{L_k}(\textbf 0)}}C(N,d)|\textbf G^{(N)}_{\Lambda^{(N)}_{L_k}}(E)(\textbf{x},\textbf{x}')|\|\textbf{x}''-\textbf{x}'\|^{-r}(1+\|\textbf{x}''\|)^{\frac{Nd}{2}+\epsilon_1}\\
			\leq& \sum_{\substack{\|\textbf{x}' \|\leq \frac{11}{20}L_k \\ \|\textbf{x}''\|>L_k}}C(\epsilon_1,N,d) L^{\frac{18^Np_0}{2\rho}}_k \Big\|\frac{9}{20}\textbf{x}''\Big\|^{-r} \|\textbf x''\|^{\frac{Nd}{2}+\epsilon_1} \\
			&\quad +  \sum_{\substack{\frac{11}{20}L_k <\|\textbf{x}' \|\leq L_k \\ \|\textbf{x}''\|>L_k}}C(\epsilon_1,N,d) \|\textbf{x}'-\textbf{x}\|^{-\frac{r_N}{20}}\|\textbf{x}''-\textbf{x}'\|^{-r}\|\textbf{x}''\|^{\frac{Nd}{2}+\epsilon_1}\\
			\leq& \sum_{\substack{\|\textbf{x}' \|\leq \frac{11}{20}L_k \\ \|\textbf{x}''\|>L_k}}C(\epsilon_1,N,d) L^{\frac{18^Np_0}{2\rho}}_k \|\textbf{x}''\|^{-r+\frac{Nd}{2}+\epsilon_1} \\
			& \quad  +  \sum_{\substack{\frac{11}{20}L_k <\|\textbf{x}' \|\leq L_k \\ \|\textbf{x}''\|>L_k}}C(\epsilon_1,N,d) \|\textbf{x}'-\textbf{x}\|^{-\frac{r_N}{20}}\|\textbf{x}''-\textbf{x}'\|^{-r}\|\textbf{x}''\|^{\frac{Nd}{2}+\epsilon_1}\\
			\leq & C(\epsilon_1,N,d)L^{\frac{18^N p_0}{2\rho}+Nd}_{k} L_k^{-\frac{r}{2}+\frac{3}{4}Nd+O(\epsilon_1)}  \\
			&\quad + \sum_{\frac{11}{20}L_k <\|\textbf{x}' \|\leq L_k }C(\epsilon_1,N,d) \|\textbf{x}'-\textbf{x}\|^{-\frac{r_N}{20}+\frac{3}{2}Nd+O(\epsilon_1) } L^{\frac{3}{2}Nd+O(\epsilon_1)}_k  \\
			\leq &C(\epsilon_1,N,d) L^{-\frac{r}{2}+\frac{18^Np_0}{2\rho}+\frac{7}{4}Nd+O(\epsilon_1)}_k  +C(\epsilon_1,N,d)L^{-\frac{r_N}{40}+2Nd+O(\epsilon_1)}_{k}\\
			\leq & C(\epsilon_1,N,d)L_k^{-\frac{31\cdot18^Np_0}{18}-4Nd+\frac{7}{4}Nd+O(\epsilon+\epsilon_1)}+C(\epsilon_1,N,d)L_k^{-\frac{r_N}{40}+2Nd+O(\epsilon_1)}\\
			\rightarrow&  0 \quad as \quad k \rightarrow \infty.
		\end{align*} 
		Hence $\bm{\phi}=0$, which contradicts the assumption. The lemma is proved.
	\end{proof}
	
	$\textbf{Proof of Theorem \ref{mainthm}}$.
	\begin{proof}
		For any {\color{red}$k \geq 0$}, define the set $A_{k}=\Lambda_{13NL_{k+1}} \backslash \Lambda_{12NL_{k}}$. From Lemma \ref{weaks}, we see that $\Lambda_{L_k}(\textbf{0})$ and any cube  $\Lambda_{L_k}(\textbf{x}), \textbf x \in A_k$ are separable. Hence, we can define the set 
		\begin{equation*}
			\textbf{E}_k=\Big\{\omega \ \Big| \ \exists E \in I, \ \text{s.t. both } \Lambda_{L_k}(\textbf{0}) \ \text{and } \ \Lambda_{L_k}(\textbf{x})\ (\textbf{x} \in A_{k} ) \text{ are } \big(E,\frac{1}{2}\big)\text{-S}\Big\}.  
		\end{equation*}
		Let $p^{(N)}=p_0=20Nd$,  one has 
		\begin{equation*}
			\mathbb{P}(\textbf{E}_k)\leq (26NL^4_k+1)^{Nd}L^{-2p_0}_{k} \leq C(N,d)L^{-36Nd}_{k},
		\end{equation*}
		
		\begin{equation}
			\sum_{k \geq 0}\mathbb{P}(\textbf{E}_k) \leq \sum_{k \geq 0}C(N,d)L^{-36Nd}_{k} <\infty.
		\end{equation}
		By the Borel-Cantelli Lemma, we can
		define the set 
		\begin{equation*}
			\Omega_0=   \big\{\omega \ | \ \omega \in  \textbf{E}_k  \text{ for finitely many }k \big\},
		\end{equation*}
		and 
		\begin{equation*}
			\mathbb{P}(\Omega_0)=1.
		\end{equation*}
		
		Let $E\in I$ be an $\epsilon_1$-generalized eigenvalue and $\bm{\psi}$ be its generalized eigenfunction, where $0 <\epsilon_1\ll 1$ will be specified later.
		From Lemma \ref{ngene},  there exist only finitely many $k$ so that $\Lambda_{L_k}(\textbf 0)$ are $(E,\frac{1}{2})$-NS. 
		
		Fix $\omega \in \Omega_0$.  Then there exists $k_0(\omega)>0$ such that for $k\geq k_0(\omega)$,  all $\Lambda_{L_k}(\textbf{x}),  \textbf{x} \in A_k $  are $(E,\frac{1}{2})$-NS. 
		Hence, for any $\textbf{x} \in A_{k}$, one has 
		
		\begin{align*}
			|\bm{\psi}(\textbf{x})| \leq &  \sum_{\substack{ \textbf{x}' \in \Lambda^{(N)}_{L_k}(\textbf x) \\  \textbf{x}'' \notin \Lambda^{(N)}_{L_k}(\textbf x)}}C(N,d)|G^{(N)}_{\Lambda^{(N)}_{L_k}}(E)(\textbf{x},\textbf{x}')|\|\textbf{x}''-\textbf{x}'\|^{-r}(1+\|\textbf{x}''\|)^{\frac{Nd}{2}+\epsilon_1} \\
			\leq & \sum_{\substack{\|\textbf{x}'-\textbf{x}\|\leq \frac{L_k}{2} \\ \|\textbf{x}''-\textbf{x}\|>L_{k}}}C(N,d)L_k^{\frac{18^N p_0}{2\rho}}\Big( \frac{1}{2}\|\textbf{x}''-\textbf{x}\|  \Big)^{-r} \big( 1+14NL_{k+1}+\|\textbf{x}''-\textbf{x}\|\big)^{\frac{Nd}{2}+\epsilon_1}\\
			&+  \sum_{\substack{ \frac{L_k}{2} < \|\textbf{x}'-\textbf{x}\|\leq L_k \\ \|\textbf{x}''-\textbf{x}\|>L_{k}}}C(N,d) \|\textbf{x}'-\textbf{x}\|^{-\frac{r_N}{20}} \|\textbf{x}''-\textbf{x}'\|^{-r} \\
			& \quad  \quad  \quad   \quad  \quad  \quad  \quad  \quad  \Big( 1+14NL_{k+1}+\|\textbf{x}''-\textbf{x}'\|\Big)^{\frac{Nd}{2}+\epsilon_1}\\
			\leq &  C(N,d,\epsilon_1)L^{\frac{18^N p_0}{2p}+Nd+2Nd+O(\epsilon_1)}_k \sum_{\|\textbf{x}''-\textbf{x}\|>L_k} \|\textbf{x}''-\textbf{x}\|^{-r+\frac{Nd}{2}+O(\epsilon_1)}\\
			&+ C(N,d, \epsilon_1)L^{3Nd+O(\epsilon_1)}_k \sum_{\|\textbf{x}'-\textbf{x}\| \geq \frac{L_k}{2}}\|\textbf{x}'-\textbf{x}\|^{-\frac{r_N}{20}}  \\
			\leq& C(N,d,\epsilon_1)\big( L_k^{-\frac{r}{2}+\frac{18^Np_0}{2\rho}+\frac{15}{4}Nd+O(\epsilon_1)}+L_k^{3Nd-\frac{r_N}{40}+\frac{Nd}{2}+O(\epsilon_1)}\big)\\
			\leq&  C(N,d,\epsilon_1)\big(L_k^{-\frac{r}{4}+\frac{7}{4}Nd+O(\epsilon_1)}+L_k^{-\frac{r}{40}+\frac{279}{80}Nd+O(\epsilon_1+\epsilon)}\big),
		\end{align*} 
		Since $r>\frac{40 \cdot 18^N p_0 }{9 \rho}+\frac{75}{9}Nd>699Nd$, and $2N^{\frac{1}{4}} L_k \geq \|\textbf{x}\|^{\frac{1}{4}} $. One has 
		\begin{equation}
			|\bm{\psi}(\textbf{x})| \leq L^{-\frac{r}{50}}_k \leq \|\textbf{x}\|^{-\frac{r}{300}}, \quad \forall \textbf{x} \in A_{k}.
		\end{equation}
		Seeing that  $\bigcup_{k\geq k_0}A_k=\{ \textbf{x} \in \mathbb{Z}^{Nd}: \|\textbf{x}\| \geq 12NL_{k_0}\}$, hence 
		$$|\bm{\psi}(\textbf{x})| \leq \|\textbf{x}\|^{-\frac{r}{300}}, \quad \forall\|\textbf{x}\|\geq 12NL_{k_0}.$$ 
		This implies the\textbf{ power-law localization} of $\mathbf{H}^{(N)}_{\omega}$ for any energy.
	\end{proof} 
	
	\section{Appendix}
	\subsection{Multi-particle Stollmann's Bound}
	
	\begin{thm} \label{wegner1}
		Consider an $n$-particle cube $\Lambda^{(n)}_L(\textbf{u})$ and suppose that the marginal probability distribution function $\mu$ of random variables $V(x,\omega)$ has the continuity modulus $\varsigma(\cdot)$:
		$$\varsigma(\epsilon):=\sup_{t \in \mathbb{R}}(\mu(t+\epsilon)-\mu(t))$$
		
		Then, $\forall E \in \mathbb{R}$, the $n$-particle power-law long-range hopping random operator $\textbf H^{(n)}_{\Lambda^{(n)}_L(\textbf{u})}$ satisfies
		\begin{align*}
			\mathbb{P}\Big\{\text{dist}\Big[E,\sigma\Big(\textbf H^{(n)}_{\Lambda^{(n)}_L(\textbf{u})}\Big)\Big]\leq \epsilon\Big\}
			\leq n(2L+1)^{(n+1)d} \varsigma(2\epsilon).
		\end{align*}	
	\end{thm}
	
	\begin{proof}
		This is a small modification of Theorem 3.4.1 in \cite{CS2014}, since the random operator $\textbf H^{(n)}_{\Lambda^{(n)}_L}(\textbf u)$ is still a family of monotone operators.
	\end{proof}
	
	We define the distance between two spectra in a usual way:
	\begin{align*}
		&\text{dist}\Big[\sigma\Big(\textbf H^{(n)}_{\Lambda^{(n)}_{L}(\textbf{u})}\Big),\sigma\Big(\textbf H^{(n)}_{\Lambda^{(n)}_{L'}(\textbf{v})}\Big)\Big]\\
		&\ \ \ \  =\min\Big\{ \Big|E^{(k)}_{\Lambda^{(n)}_{L}(\textbf{u})}-E^{(k')}_{\Lambda^{(n)}_{L'}(\textbf{v})}\Big|:\ 1\leq k\leq \big|\Lambda^{(n)}_{L}(\textbf{u})\big|,\ 1\leq k'\leq \big|\Lambda^{(n)}_{L'}(\textbf{v})\big|   \Big\}.
	\end{align*}
	\begin{thm}[Two-volume Stollmann's Bound]\label{wegnern}
		Let $\Lambda^{(n)}_L(\textbf{u}),\ \Lambda^{(n)}_L(\textbf{v})$ be  a pair of $n$-particle weakly separable cubes. If the marginal probability distribution function $\mu$ of the random variables $V(x,\omega)$ has the continuity modulus $\varsigma(\cdot)$, then the $n$-particle power-law long-range hopping random operator $\textbf H^{(n)}_{\Lambda^{(n)}_L(\textbf{u})}$ and  $\textbf H^{(n)}_{\Lambda^{(n)}_L(\textbf{v})}$ satisfy
		\begin{align*}
			&\mathbb{P}\Big\{\text{dist}\Big[\sigma\Big(\textbf H^{(n)}_{\Lambda^{(n)}_L(\textbf{u})}\Big),\sigma\Big(\textbf H^{(n)}_{\Lambda^{(n)}_{L}(\textbf{v})}\Big)\Big]\leq \epsilon\Big\}\leq n(2L+1)^{(2n+1)d} \varsigma(2\epsilon).
		\end{align*}	
	\end{thm}
	\begin{proof}
		This is a small modification of Theorem 3.4.2 in \cite{CS2014}.
	\end{proof}

	\subsection{ Coupling Lemma}\
	
	Assume the following relations hold true:
	\begin{align}\label{condition}
		\left\{	\begin{array}{cc}
			-\frac{1}{2} r_n+\tau_n+2s^{(n)}_0<0,\\
			-  r_n+\tau_n+4\beta+\frac{15}{2}s_0^{(n)}<0,\\
			\frac{\tau_n}{2}+2\beta+\frac{7}{2}s^{(n)}_0<\tau_n,
		\end{array}\right.
	\end{align}
	where $ s^{(n)}_0>nd/2$.
	Assume 
	\begin{align}{\label{tau}}
		\beta\geq18^Np_0/ 2\rho, \ p_0\geq20Nd.  
	\end{align}

	\begin{lem}[Coupling Lemma]\label{coupling}
		Let $L=[l^{4}]$. Assume the following:
		\begin{itemize}
			\item The relation (\ref{condition}) hold true.
			\item  We can decompose the n-particle cube $\Lambda^{(n)}_{L}({\textbf u})$ into two disjoint subjects $\Lambda^{(n)}_L(\textbf u)=B\cup G$ with the following properties: We have 
			$$B=\cup_{1\leq j< +\infty}\Omega_j,$$
			where for each $j$, $\text{diam}(\Omega_j)\leq C_{*}l^{2}$    $(C_{*}>1)$, and for $j\neq j'$, dist$(\Omega_j,\Omega_{j'})\geq l^{2}$. For any ${\textbf v}\in G$, there exists some $\Lambda^{(n)}_l (\tilde{\textbf v})\subset \Lambda^{(n)}_L({\textbf u})$ such that $\Lambda^{(n)}_l(\tilde{\textbf v})$ is $(E,\frac{1}{2})$-NS and $\textbf v\in\Lambda^{(n)}_l (\tilde{\textbf v})$ with dist$\left(\textbf v,\Lambda^{(n)}_L({\textbf u})\setminus\Lambda^{(n)}_l(\tilde{\textbf v})\right)\geq l/2$.
			\item $\Lambda^{(n)}_L(\textbf u)$ is $E$-NR.
		\end{itemize}	
		Then for $$l\geq \underline{l}_0(\|\mathbf{T}+\mathbf{U}\|_{r_n}, M, C_{*},\beta,\tau_n,r_n, s^{(n)}_0,n,d)>0,$$ we have $\Lambda^{(n)}_{L}(\textbf u)$ is $(E,\frac{1}{2})$-NS.
	\end{lem}	
	\begin{proof}
		This is a modification of Lemma 3.2 in \cite{Shi2021} by taking $\xi=1,\ \alpha=4$ and $\delta=\frac{1}{2}$.
	\end{proof}
	
	\subsection{The Separation Property  of Singular Cubes}
	\begin{lem}\label{SC}
		Fix $\mathcal{J} \in \mathbb{N}$. Assume there exists $\mathcal{J} $ cubes of size $\tilde{C}(n)l$ such that for any $\textbf{v}$ does not belong to these cubes,  $\Lambda^{(n)}_{l}(\textbf{v})$ is $(E,\frac{1}{2})$-NS. Then for 
		$$l >2\tilde{C}(n),$$
		the n-particle cube $\Lambda^{(n)}_{L}({\textbf u})$ can be decomposed into two disjoint subjects $B\cup G$ with the following properties: we have 
		$$B=\cup_{1\leq j< +\infty}\Omega_j,$$
		where for each $j$, $\text{diam}(\Omega_j)\leq (2\mathcal{J}+1)l^{2}$, and for $j\neq j'$, dist$(\Omega_j,\Omega_{j'})\geq l^{2}$. For any ${\textbf v}\in G$,  $\Lambda^{(n)}_{l}(\textbf{v})$ is $\big(E,\frac{1}{2}\big)$-NS.
	\end{lem}
	\begin{proof}
		Denote $\Lambda^{(n)}_{\tilde{C}(n)l}(\mathbf k^{(1)}),\cdots, \Lambda^{(n)}_{\tilde{C}(n)l}(\mathbf k^{(\mathcal{J})})$ as these $n$-particle cubes of size $\tilde{C}(n)l$, and define $\mathrm{Z}=\{\mathbf k^{(1)},\cdots,\mathbf k^{(\mathcal{J})}\}$. We can define a relation $\eqcirc$ on $\mathrm{Z}$ as follows. Letting $\mathbf k,\mathbf k' \in \mathrm{Z}$, we say that $\mathbf k\eqcirc \mathbf k'$, if there exists a sequence $\mathbf k_0,\cdots, \mathbf k_q \in \mathrm{Z}$, such that 
		$$\mathbf k=\mathbf k_0,\quad \mathbf k'=\mathbf k_q,$$
		and 
		\begin{equation}
			|\mathbf k_i-\mathbf k_{i+1}| \leq 2 l^2, \quad \forall 0\leq i\leq q-1.  
		\end{equation}
		As a consequence,  the set $\mathrm{Z}$  can be divided into disjoint equivalent classes, say $\chi_1,\cdots \chi_{t}$ with $t \leq \mathcal{J}$. We can obtain
		\begin{equation}\label{sep1}
			|\mathbf k-\mathbf k'|\leq 2\mathcal{J}l^2, \quad \forall\ \mathbf k,\mathbf k' \in \chi_{j}
		\end{equation}
		\begin{equation}\label{sep2}
			\text{dist}(\chi_i,\chi_j)>2l^2,\quad  \text{for}\ i\neq j.
		\end{equation}
		Correspondingly, we can deﬁne
		\begin{equation*}
			\Omega_j=\bigcup_{\mathbf k \in \chi_j}\big(\Lambda^{(n)}_{L}(\textbf{u}) \cap \Lambda^{(n)}_{\tilde{C}(n)l}(\mathbf k)\big).
		\end{equation*}
		From \eqref{sep1} and \eqref{sep2}, we see 
		\begin{equation*}
			\text{diam}(\Omega_j) \leq 2Jl^2+2\tilde{C}(n)l \leq (2J+1)l^2 
		\end{equation*}
		\begin{equation*}
			\text{dist}(\Omega_i,\Omega_j) >2l^2-2\tilde{C}(n)l>l^2, \ \text{for} \ i\neq j.
		\end{equation*}
		Let $B=\cup _{1\leq j\leq t}\Omega_j$, once $\textbf{v} \in \Lambda^{(n)}_{L}(\textbf{u}) \setminus B$, we see that $\Lambda^{(n)}_{l}(\textbf{v})$ is $(E,\frac{1}{2})$-NS.
	\end{proof}
	\subsection{The Summation of Sequence}
	\begin{lem}[Lemma A.1 of \cite{Shi2021}]\label{sum}
		Let $L>2$ with $L\in\N$ and $\Theta-Nd>1$. Then we have 
		\begin{align*}
			\sum_{\textbf u\in\Z^{Nd}:|\textbf u|\geq L}|\textbf u|^{-\Theta}\leq C(\Theta,Nd)L^{-(\Theta-Nd)/2},
		\end{align*}
		where $C(\Theta,Nd)$ depends only on $\Theta,\ N$ and $d$.
	\end{lem}

	\section{Acknowledgements}
	Our research is supported by the NSFC (No. 12101542; No. 12201392) and Science and Technology Commission of Shanghai Municipality (No. 22YF1414100).

	\section{Declarations}
	
	\textbf{Conflict of interest} \quad On behalf of all authors, the corresponding author states that there is no conflict of interest.

	\textbf{Data Availability} \quad 
	The manuscript has no associated data.
	
	\bibliographystyle{alpha} 
	\bibliography{KAM}
	
\end{document}